\documentclass[11pt]{article}

\pdfoutput=1 

\usepackage[utf8]{inputenc}
\usepackage[left=1in, right=1in, top=1in, bottom=1in]{geometry}

\usepackage[title]{appendix}

\usepackage[usenames,dvipsnames,svgnames,table]{xcolor}
\definecolor{NewBlue}{rgb}{0.2,0.2,0.8}

\usepackage[
  ruled,          
  vlined,         
  linesnumbered,  
  commentsnumbered,
  inoutnumbered,
  titlenumbered
]{algorithm2e}
\usepackage{authblk}
\usepackage{tcolorbox}

\newtcolorbox{boxtext}{
  colback=gray!10,
  colframe=black,
  boxrule=0.5pt,
  arc=2pt,
  left=6pt,
  right=6pt,
  top=6pt,
  bottom=6pt
}

\usepackage{caption}
\usepackage{subcaption}
\usepackage{subfiles}
\usepackage{tikz}
\usepackage{comment}
\usetikzlibrary{arrows.meta}
\usetikzlibrary{arrows, automata}
\usetikzlibrary{calc}
\usetikzlibrary{shapes.geometric}
\usepackage{amssymb,amsmath,amsthm}
\usepackage{bbm}
\usepackage{xspace}
\usepackage{hyperref}

\newtheorem{definition}{Definition}
\newtheorem{theorem}{Theorem}
\newtheorem{corollary}{Corollary}
 \newtheorem{lemma}{Lemma}

\usepackage[capitalise]{cleveref}
\Crefname{equation}{Eq.}{Eqs.}
\Crefname{figure}{Fig.}{Figs.}
\Crefname{tabular}{Tab.}{Tabs.}

\DeclareMathOperator\tw{tw} 

\graphicspath{{./images/}}

\newcommand{\myAlgorithm}{\textsc{TreeTrimmer}\xspace}
\newcommand{\ourAlgorithm}{\textsc{HedgeTrimmer}\xspace}

\usepackage{etoolbox}
\usepackage{appendix}
\AtBeginEnvironment{appendices}{\crefalias{section}{appendix}}

\date{\vspace{-5ex}}
\title{Efficient Algorithms for Interdicting Facilities in Trees and Bounded Treewidth Graphs\footnote{This work was supported in part by the U.S. National Science Foundation (NSF) CNS-1816887 award.}} 

\author{Ali Abbasi\footnote{University of Southern California, Department of Computer Science, Los Angeles, CA 90089. \href{mailto:abbasia@usc.edu}{\texttt{abbasia@usc.edu}}.}
\qquad
Eli Friedman\footnote{University of Utah, Department of Computer Science, Salt Lake City, UT 84112. \href{mailto:eli.friedman@utah.edu}{\texttt{eli.friedman@utah.edu}}. Part of this work was done while an intern at the University of Southern California.}
\qquad
Leana Golubchik\footnote{University of Southern California, Department of Electrical and Computer Engineering, Los Angeles, CA 90089. \{\href{mailto:leana@usc.edu}{\texttt{leana}},\href{mailto:paolieri@usc.edu}{\texttt{paolieri}}\}\texttt{@usc.edu}.}
\qquad
Samir Khuller\footnote{Northwestern University, Department of Computer Science, Evanston, IL 60208. \href{mailto:samir.khuller@northwestern.edu}{\texttt{samir.khuller@northwestern.edu}}.}
\\
Marco Paolieri$^\S$
}

\hypersetup{colorlinks,linkcolor={NewBlue},citecolor={NewBlue},urlcolor={black}}  

\makeatletter
\def\@footnotecolor{black}
\define@key{Hyp}{footnotecolor}{%
 \HyColor@HyperrefColor{#1}\@footnotecolor%
}
\patchcmd{\@footnotemark}{\hyper@linkstart{link}}{\hyper@linkstart{footnote}}{}{}
\makeatother
\hypersetup{footnotecolor=black}

\begin{document}
\maketitle
\begin{abstract}

Given a graph $G$ of $n$ nodes partitioned into facilities and customers, the \emph{$r$-edge interdiction covering problem} (REIC) is to remove up to $r$ edges so as to maximize the total weight of customers disconnected from all facilities, which is called the covering objective function. 
While REIC is known to be NP-complete for general graphs, Fröhlich and Ruzika show that the problem can be solved in polynomial time when $G$ is a tree, providing an $O(n^7 r)$-time algorithm.
%
We give an efficient $O(nr^2)$-time dynamic programming algorithm for REIC on trees that is fixed-parameter linear in $n$.
Evaluating our solution on a benchmark of randomly generated tree networks with baselines of the Fröhlich and Ruzika algorithm and the Gurobi integer program solver, we demonstrate that in practice, our algorithm is both significantly faster and less sensitive to network topology and size.

We extend our algorithm for REIC to graphs of bounded treewidth, a well-studied family of sparse graphs that generalizes trees, and obtain a matching runtime of $O(nr^2)$.
We also consider the \emph{$r$-facility interdiction covering problem} (RFIC), a novel variant of this network interdiction problem  where the goal is to remove up to $r$ facilities to maximize the covering objective function over disconnected customers.
We show that RFIC is NP-complete by observing it generalizes the small set bipartite vertex expansion problem (SSBVE), also known as the minimum $p$-union problem.
We give an $O(nr^2)$-time algorithm for RFIC on trees, which also gives an $O(n^3)$-time algorithm for SSBVE on trees.

\end{abstract}

\newpage

\section{Introduction}\label{sec:intro}

The problem of network interdiction has a rich history rooted in the famous max-flow min-cut theorem~\cite{FordF56}, which established the relationship between the network flow capacity and the critical set of edges that support it.
Later, the problem of interdicting the maximum flow with a limited budget of edge removals was addressed~\cite{Wollmer64,Phillips93, Burch2003}, which is known as the \emph{network flow interdiction problem}\footnote{We are given a flow network where each edge is associated with both a cost and a capacity. Given a budget $B$, we need to delete edges of total cost at most $B$ to maximally reduce the max $s$-$t$ flow.}~(NFI).
In addition to maximum flow, many other network properties have been considered as the objective of network interdiction.
These include minimum weight spanning trees~\cite{LiangS97,FredericksonS99, BazganTree}, maximum weight matchings~\cite{Zenklusen10a}, and shortest paths between two nodes~\cite{FulkersonH77, CorleyS82, IsraeliW02, bar-noy2001}.

Location science also considers related optimization problems to evaluate the robustness of facility locations in supply-demand networks~\cite{SnyderAPRSS16}.
In the placement of $p$ facilities, problems of interest include the \emph{$p$-median problem}~\cite{BazganPmedian,Hakimi64pmedian}, i.e., finding facilities that minimize the sum of distances to customer nodes, and the \emph{$p$-covering problem}~\cite{church1974maximal}, i.e., maximizing the total demand of customer nodes within the covering radius of a facility.
Previous work proposes the corresponding interdiction problems, respectively, the \emph{$r$-interdiction median problem}~(RIM) and the \emph{$r$-interdiction covering problem}~(RIC), as a way to assess the robustness of facility locations when $r$~facilities can be removed adversarially~\cite{ChurchSM04}.
Further works in this area explore facility removal in hub networks~\cite{RamamoorthyJSV18,GhaffariA18} as well as fortification against facility removal.

In contrast, we consider a network interdiction problem where edges are removed to prevent customer nodes from reaching facilities, which we call the \emph{$r$-edge interdiction covering problem}~(REIC).
We now define the problem, letting $C_G(v,R)=1$ indicate that customer $v \in V\setminus S$ can reach some facility $s\in S$ in the graph $G\setminus R$ for \emph{interdiction strategy} $R\subset E$.
\begin{definition}[REIC]
    In the $r$-edge interdiction covering problem, we are given a graph $G=(V,E)$ with facilities $S\subset V$, customer weights $w:V\setminus S\rightarrow \mathbb{R}_{\geq 0}$, and an integer budget $r<|E|$, with the goal of selecting a subset $R\subset E$ with $|R|\leq r$ maximizing $\sum_{v\in V\setminus S}w(v)(1-C_G(v,R))$.
\end{definition}
Thus, the goal is to remove $r$ edges to minimize the \emph{covering objective function}. By adding a super source and sink connected to all facilities and customers, respectively, we can view this problem as a special case of network flow interdiction.\footnote{Set interdiction cost to be $1$ for original edges in the graph, $\infty$ otherwise. Set capacity of customer-sink edges to be $w(v)$ for customer $v$, $\infty$ otherwise.}
However, NFI is hard even to approximate well~\cite{Phillips93,ChestnutZenklusen2017}, with best known approximation ratio of $2(n-1)$~\cite{ChestnutZenklusen2017}.

While $r$-edge interdiction covering has received less attention than other interdiction problems, for general graphs, Fröhlich and Ruzika~\cite{FrohlichR21} showed that REIC is NP-complete.
For trees, the same authors provided a polynomial-time algorithm~\cite{FrohlichR22}, dividing the graph into smaller components where facilities appear only on leaves and then combining partial solutions with a knapsack-type approach.
The exact runtime is not specified, but on a tree with $n$ nodes and budget $r$, our analysis suggests that the worst-case running time is $O(n^7r)$.\footnote{This occurs when a connected cluster of customers has size $O(n)$ and each customer is attached to a facility.}

In this paper, we improve this runtime to $O(nr^2)$ with \myAlgorithm, an algorithm based on a knapsack-type dynamic programming (DP) approach, processing the input tree bottom-up and solving a multiple-choice knapsack problem~\cite{Kellerer04} at each node. 
Our algorithm is fixed-parameter linear (FPL, \cite{Cygan2015ParameterizedA}) in the parameter $r$, meaning its time complexity is linear in $n$ when the budget $r$ is constant.
We then extend our DP approach for $r$-edge interdiction covering to bounded treewidth graphs, a well-studied family of sparse graphs. 
Informally, treewidth~$k$ is a measure of how far a graph is from being a tree.
Our algorithm \ourAlgorithm  is fixed-parameter tractable (FPT, \cite{Cygan2015ParameterizedA}) in the parameter $k$, running in time $O(nr^2)$ when treewidth is bounded by a constant, thus matching our result for trees.

We also consider a formulation of RIC on graphs, which we call the \emph{$r$-facility interdiction covering problem}~(RFIC).
This problem is similar to REIC, except that the interdiction strategy removes facilities rather than edges from the graph.
Modifying our DP algorithm for REIC on trees, we obtain an $O(nr^2)$-time algorithm for $r$-facility interdiction covering on trees.
Additionally, we observe that RFIC can be viewed as a weighted generalization of the \emph{small set bipartite vertex expansion problem}\footnote{In this problem, the goal is to select a subset $S$ of $k$ vertices from the left side of a bipartite graph minimizing the neighborhood's cardinality $|N(S)|$. In \cite{SSBVE2}, the authors show SSBVE is also equivalent to the minimum $p$-union problem.} (SSBVE). 
An $O(\sqrt{n})$-approximation for SSBVE was presented in~\cite{SSBVE1} and then improved to an $O(n^{1/4+\varepsilon})$-approximation in~\cite{SSBVE2}, which is believed to be optimal under the ``Dense versus Random'' conjecture for the densest $k$-subgraph and smallest $p$-edge subgraph problems~\cite{SSBVE2,BhaskaraHardness, Manurangsi}.
%
%
In this work, we give an approximation-preserving reduction from SSBVE to RFIC, showing that RFIC is
as hard to approximate as SSBVE and hence smallest $p$-edge subgraph (when viewed as minimizing the weight of covered customers);
this reduction also gives an $O(n^3)$-time exact algorithm for SSBVE on trees.
%
Finally, we give a parameterized reduction directly from $\textsc{Clique}$ to show that RFIC is W[1]-hard and as hard to approximate as densest $k$-subgraph (when viewed as maximizing the weight of disconnected customers).

One interesting application that motivates an efficient solution to this interdiction problem is robust network design of modern Low Earth Orbit (LEO) satellite constellations (e.g., Starlink \cite{starlink} and Amazon Kuiper \cite{kuiper}) that are capable of establishing intersatellite links (ISLs).
In this setting, facilities model ground stations and customers model satellites receiving data from users or collecting remote sensing data for Earth observation.
A customer's weight captures the amount or value of data received by a satellite, which needs to be delivered to ground stations for routing (in case of broadband access) or analysis (in case of remote sensing); here, our covering interdiction problems quantify the vulnerability of a constellation, maximizing the total amount of data that can be prevented from reaching the ground.
%
%
Each satellite can maintain only a few ISLs (e.g., Starlink supports 4), used to connect with other satellites on the same orbital plane (e.g., the preceding and following satellites) or on nearby orbital planes (e.g., moving in the same direction on the closest orbital planes) \cite{ChaudhryY21}.
Thus, ISLs result in very sparse graphs, and multicast spanning trees can be used to monitor and control applications for satellite communication \cite{EkiciAB02,BonuccelliMP04,HuXXDDP22}.
While there is recent literature on the vulnerability of satellite constellations to various attack or failure scenarios \cite{XuGL23,WangLLL24}, given their sparsity, we believe that our algorithms for trees and bounded treewidth graph can be used to efficiently evaluate their robustness to interdiction.

Initially, we considered satellite constellations with only 2 ISLs, which form graphs of paths and cycles; these simple graphs led us to our algorithm for trees.
While it is not practical to assume satellite ISLs form trees, we believe that our bounded treewidth assumption is well-motivated by the Walker Delta Constellation \cite{walker,mclaughlin2023xgrid}, where satellites are organized into $k$ orbital planes and maintain intra-plane ISLs to their leading and trailing neighbors. This forms a ring when the plane is fully populated, or a path if the satellites do not fully encircle the Earth. Adjacent orbital planes also maintain inter-plane ISLs: each satellite connects to (at most) one satellite in each neighboring plane, forming a sparse grid-based topology in which each satellite has degree at most $4$. Ground stations can then be modeled as additional nodes that attach to satellites currently in view, contributing leaf-like connections without changing the constellation's grid structure. Under this grid-based abstraction of the Walker Delta Constellation, the network resembles a $k \times n$ grid, where $n$ is the number of satellites in each plane and $k$ is the number of planes. Since a $k \times n$ grid has treewidth $\min(k,n)$, the treewidth of the constellation graph scales only with $k$ for $k<n$, and hence our algorithms can be used in practice when the number of orbital planes is small.

%
%

In summary, \emph{our main contributions} are as follows.
\begin{itemize}
    \item In \cref{sec:edge}, we improve upon the $O(n^7r)$-time algorithm of Fr{\"o}hlich and Ruzika \cite{FrohlichR22} and provide an $O(nr^2)$-time DP algorithm for solving $r$-edge interdiction covering on trees, which is FPL in the budget~$r$. We then evaluate our algorithm for REIC on trees in \cref{sec:experimental}, using the Gurobi integer program solver and \cite{FrohlichR22} as baselines across randomly generated problem instances. Our numerical results demonstrate that our algorithm is significantly faster in practice and more robust to variations in tree and instance structure.
    \item \Cref{sec:bddtw} extends our algorithm for REIC on trees to give an FPT algorithm for bounded treewidth graphs, achieving the same $O(nr^2)$-time complexity when the treewidth is constant. We also compare our results to the meta-algorithm of Courcelle's theorem and highlight how our DP algorithms are less constrained in instance structure while also giving strong evidence favoring our algorithms' performance in practice.
    \item In \cref{sec:facility}, we propose $r$-facility interdiction covering, a novel variant with facility rather than edge removal. \cref{sec:facility-DP} modifies our algorithm for REIC on trees to give an $O(nr^2)$-time algorithm for RFIC on trees. We give a reduction from small set bipartite vertex expansion to RFIC in \cref{sec:SSBVE}, resulting in an $O(nk^2)$-time algorithm for SSBVE on trees and hardness of approximation for RFIC-min. We show that RFIC is W[1]-hard in \cref{sec:hardness}, along with hardness of approximation for RFIC-max.
\end{itemize}

\section{TreeTrimmer: a DP Algorithm for REIC on Trees}\label{sec:edge}
Let $T=(V,E)$ be a tree of size $|V|=n$ with facilities $S\subset V$ and non-negative customer weights $w$.

For solving REIC on $T$ with budget $r$, we propose \myAlgorithm, a DP algorithm which traverses the tree upwards from leaves towards an arbitrarily chosen root $v^*$. At each inner node $v$, our algorithm maximizes the total weight of customers disconnected in the subtree rooted at $v$. It does so under every combination of two binary conditions: an \emph{assumption} $X$ provided by the parent node and a \emph{guarantee} $Y$ enforced in the subtree rooted at $v$.

Specifically, we define two possible assumptions $X$ regarding $v$'s parent node $w$:
\begin{itemize}
	\item $X=0$ when \emph{no} facility is reachable from $v$ through $w$;
	\item $X=1$ when \emph{a} facility is reachable from $v$ through $w$.
\end{itemize}
As for condition $Y$, if we consider only the subtree rooted at $v$ (i.e., without traversing the parent $w$ to reach facilities), two guarantees can be enforced through edge removals:
\begin{itemize}
    \item $Y=0$ when \emph{no} facility is reachable from $v$ in the subtree of $v$;
    \item $Y=1$ when \emph{a} facility is reachable from $v$ in the subtree of $v$;
\end{itemize}
We denote by $V_{XY}(v,b)$ the maximum objective (total weight of disconnected customers) achievable in the subtree rooted at $v$ under the conditions $XY$ with budget $b$,
set to $-\infty$ when no feasible solution exists for the given conditions, budget, and node.
Note that the $XY$'s are exhaustive: any interdiction strategy $R$ inherently fixes conditions $XY$ and budget~$b$ for each node and its subtree. 
Thus, by taking the maximum value of $V_{XY}(v^*,r)$ for all conditions $XY\in\{00,01,10,11\}$, we find a solution with optimal interdiction value.

\def\scalefactor{0.65}
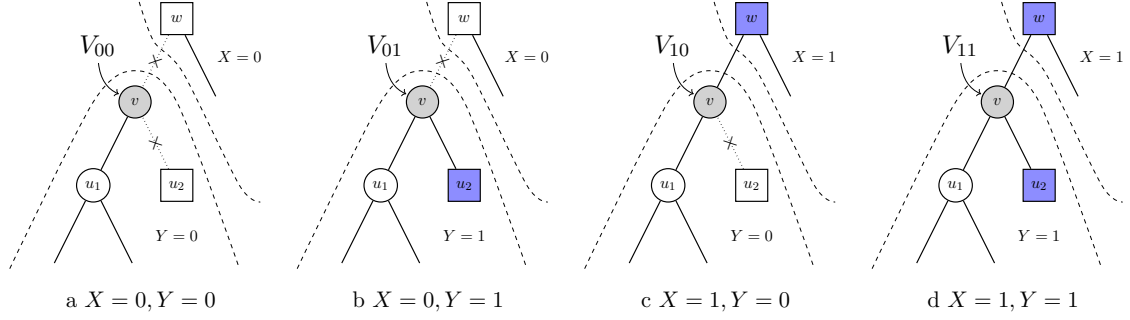
\begin{figure*}[t]
\centering
\captionsetup[subfloat]{labelformat=simple}
\scalebox{0.85}{%
\subfloat[$X=0, Y=0$\label{fig:1a}]{%
    \scalebox{\scalefactor}{ 
\begin{tikzpicture} [square/.style={regular polygon,regular polygon sides=4}]
\begin{scope}[every node/.style={circle,thick,draw,minimum size = 2 em}]
    \node (u1) at (2,1) {$u_1$};
    \node [fill = gray, fill opacity = 0.35, text opacity=1] (v) at (3,3) {$v$};
\end{scope}
\begin{scope}[every node/.style={square,thick,draw,minimum size = 2 em}]
    \node (w) at (4,5) {$w$};
    \node (u2) at (4,1) {\hphantom{$w$}};
\end{scope}
\node (z1) at (1,-1) {};
\node (z2) at (3,-1) {};
\node (v2) at (5,3) {};
\node (u2label) at (u2) {$u_2$};

\node [scale = 1.6] (label) at (2.1,4.35) {$V_{00}$};
\path [->, shorten <= -0.15em, shorten >=0.1em] (label) edge[bend right=30] (v);

\node (comment1) at (5.5,4.1) {$X=0$};

\node (comment2) at ($(u2)+(0,-1.2)$) {$Y=0$};

\begin{scope}[every edge/.style={draw=black, thick}]
    \path [-] (u1) edge (v);
    
    \path [dotted] (w) edge node [sloped] {\Large $\times$} (v);
    
    \path [dotted] (u2) edge node [sloped] {\Large $\times$} (v);

    \path [-] (u1) edge (z1);
    \path [-] (u1) edge (z2);
    \path [-] (w) edge (v2);
        
    \draw [dashed] plot [smooth] coordinates {(0,-1)  (2,3)  (3,3.75)  (4,3) (5.5,-1)};
     \draw [dashed] plot [smooth] coordinates {(3.1,5.4)  (3.45,4.45)  (4.15,3.7)  (5.5,1) (6,0.6)};
\end{scope}

\end{tikzpicture}}}%
\hspace{1em}%
\subfloat[$X=0, Y=1$\label{fig:1b}]{%
    \scalebox{\scalefactor}{\begin{tikzpicture} [square/.style={regular polygon,regular polygon sides=4}]
\begin{scope}[every node/.style={circle,thick,draw,minimum size = 2 em}]
    \node (u1) at (2,1) {$u_1$};
    \node [fill = gray, fill opacity = 0.35, text opacity=1] (v) at (3,3) {$v$};
\end{scope}
\begin{scope}[every node/.style={square,thick,draw,minimum size = 2 em}]
    \node (w) at (4,5) {$w$};
    \node [fill = blue, fill opacity = 0.45, text opacity=1] (u2) at (4,1) {\hphantom{$w$}};
\end{scope}
\node (z1) at (1,-1) {};
\node (z2) at (3,-1) {};
\node (v2) at (5,3) {};
\node (u2label) at (u2) {$u_2$};

\node [scale = 1.6] (label) at (2.1,4.35) {$V_{01}$};
\path [->, shorten <= -0.15em, shorten >=0.1em] (label) edge[bend right=30] (v);

\node (comment1) at (5.5,4.1) {$X=0$};

\node (comment2) at ($(u2)+(0,-1.2)$) {$Y=1$};

\begin{scope}[every edge/.style={draw=black, thick}]
    \path [-] (u1) edge (v);
    
    \path [dotted] (w) edge node [sloped] {\Large $\times$} (v);
    
    \path [-] (u2) edge (v);

    \path [-] (u1) edge (z1);
    \path [-] (u1) edge (z2);
    \path [-] (w) edge (v2);
        
    \draw [dashed] plot [smooth] coordinates {(0,-1)  (2,3)  (3,3.75)  (4,3) (5.5,-1)};
     \draw [dashed] plot [smooth] coordinates {(3.1,5.4)  (3.45,4.45)  (4.15,3.7)  (5.5,1) (6,0.6)};
\end{scope}

\end{tikzpicture}}}%
\hspace{1em}%
\subfloat[$X=1, Y=0$\label{fig:1c}]{%
    \scalebox{\scalefactor}{\begin{tikzpicture} [square/.style={regular polygon,regular polygon sides=4}]
\begin{scope}[every node/.style={circle,thick,draw,minimum size = 2 em}]
    \node (u1) at (2,1) {$u_1$};
    \node [fill = gray, fill opacity = 0.35, text opacity=1] (v) at (3,3) {$v$};
\end{scope}
\begin{scope}[every node/.style={square,thick,draw,minimum size = 2 em}]
    \node [fill = blue, fill opacity = 0.45, text opacity=1] (w) at (4,5) {$w$};
    \node (u2) at (4,1) {\hphantom{$w$}};
\end{scope}
\node (z1) at (1,-1) {};
\node (z2) at (3,-1) {};
\node (v2) at (5,3) {};
\node (u2label) at (u2) {$u_2$};

\node [scale = 1.6] (label) at (2.1,4.35) {$V_{10}$};
\path [->, shorten <= -0.15em, shorten >=0.1em] (label) edge[bend right=30] (v);

\node (comment1) at (5.5,4.1) {$X=1$};

\node (comment2) at ($(u2)+(0,-1.2)$) {$Y=0$};

\begin{scope}[every edge/.style={draw=black, thick}]
    \path [-] (u1) edge (v);
    
    \path [-] (w) edge (v);
    
    \path [dotted] (u2) edge node [sloped] {\Large $\times$} (v);

    \path [-] (u1) edge (z1);
    \path [-] (u1) edge (z2);
    \path [-] (w) edge (v2);
        
    \draw [dashed] plot [smooth] coordinates {(0,-1)  (2,3)  (3,3.75)  (4,3) (5.5,-1)};
     \draw [dashed] plot [smooth] coordinates {(3.1,5.4)  (3.45,4.45)  (4.15,3.7)  (5.5,1) (6,0.6)};
\end{scope}

\end{tikzpicture}}}%
\hspace{1em}%
\subfloat[$X=1, Y=1$\label{fig:1d}]{%
    \scalebox{\scalefactor}{\begin{tikzpicture} [square/.style={regular polygon,regular polygon sides=4}]
\begin{scope}[every node/.style={circle,thick,draw,minimum size = 2 em}]
    \node (u1) at (2,1) {$u_1$};
    \node [fill = gray, fill opacity = 0.35, text opacity=1] (v) at (3,3) {$v$};
\end{scope}
\begin{scope}[every node/.style={square,thick,draw,minimum size = 2 em}]
    \node [fill = blue, fill opacity = 0.45, text opacity=1] (w) at (4,5) {$w$};
    \node [fill = blue, fill opacity = 0.45, text opacity=1] (u2) at (4,1) {\hphantom{$w$}};
\end{scope}
\node (z1) at (1,-1) {};
\node (z2) at (3,-1) {};
\node (v2) at (5,3) {};
\node (u2label) at (u2) {$u_2$};

\node [scale = 1.6] (label) at (2.1,4.35) {$V_{11}$};
\path [->, shorten <= -0.15em, shorten >=0.1em] (label) edge[bend right=30] (v);

\node (comment1) at (5.5,4.1) {$X=1$};

\node (comment2) at ($(u2)+(0,-1.2)$) {$Y=1$};

\begin{scope}[every edge/.style={draw=black, thick}]
    \path [-] (u1) edge (v);
    
    \path [-] (w) edge (v);
    
    \path [-] (u2) edge (v);

    \path [-] (u1) edge (z1);
    \path [-] (u1) edge (z2);
    \path [-] (w) edge (v2);
        
    \draw [dashed] plot [smooth] coordinates {(0,-1)  (2,3)  (3,3.75)  (4,3) (5.5,-1)};
     \draw [dashed] plot [smooth] coordinates {(3.1,5.4)  (3.45,4.45)  (4.15,3.7)  (5.5,1) (6,0.6)};
\end{scope}

\end{tikzpicture}}}%
}
\caption{Each case is $XY$ value for node $v$ under a different interdiction example on the same graph; customers are circles; facilities are squares; facilities reachable from $v$ are shaded in blue.}
\label{fig:1}
\end{figure*}

We now give the DP equations that define \myAlgorithm.

\subsection{Base Case: Leaves}

We start by defining values $V_{XY}(v,b)$ for when $v$ is a leaf of the tree with parent $w$.
\begin{itemize}
    \item For a facility leaf node $v \in S$, regardless of whether a facility can be reached through~$w$ (i.e., for both $X=0$ and $X=1$), we set $V_{X1}(v,b) = 0$ as $v$'s subtree has no customers to disconnect and $V_{X0}(v,b)=-\infty$ since facility $v$ is always reachable in its subtree.
    Specifically, for each $b \leq r$,
    \begin{align*}
    V_{XY}(v,b)=\begin{cases}
        -\infty & \text{if $Y=0$},\\
        0       & \text{if $Y=1$}.
    \end{cases}
    \end{align*}

    \item For a customer leaf node $v \in V\setminus S$, we set $V_{00}(v,b) = w(v)$ and $V_{10}(v,b) = 0$ because $v$ is connected to a facility if and only if a facility is reachable through $w$ (i.e., when $X=1$). On the other hand, a facility is never reachable in the subtree rooted at $v$ regardless of $X$, so $V_{X1}(v,b)=-\infty$. Specifically, for each $b \leq r$,
    \begin{align*}
    V_{XY}(v,b)=\begin{cases}
        w(v) & \text{if $X=0$ and $Y=0$},\\
        0       & \text{if $X=1$ and $Y=0$},\\
        -\infty & \text{if $Y=1$}.
    \end{cases}
    \end{align*}
\end{itemize}

\subsection{Recursive Case: Inner Nodes}
With DP values calculated for each leaf, we now consider an inner node $v$ with children $u_1,\dots,u_k$ and calculate $V_{XY}(v,b)$ for all budgets $b\leq r$ and cases $XY\in \{00,01,10,11\}$. In particular, $V_{XY}(v,b)$ is evaluated by solving a knapsack problem on the table of DP values of its children, $V_{XY}(u_i,b_i)$ for $i=1,\dots,k$. Here, budget $b$ is partitioned into budgets $b_i$ for the subtree of each $u_i$, such that $\sum_{i=1}^{k}b_i=b$.

\subparagraph*{Inner Nodes: Case $V_{00}$}

If $v$ itself is a facility, it is not possible to prevent $v$ from reaching a facility in the subtree rooted at $v$ as required by guarantee $Y=0$, so we set $V_{00}(v,b) = -\infty$ for all $b\leq r$.

If $v$ is a customer, then we claim for all $b\leq r$,
\begin{align}\label{eq:V00}
V_{00}(v,b) = w(v) + \max_{b_i:\sum_{i=1}^k b_i = b} \Bigg\{ \;\sum_{i=1}^k\;\max\bigg\{
V_{00}(u_i,b_i),
V_{01}(u_i,b_i-1)
\bigg\} \Bigg\} \,.
\end{align}

\begin{proof}
Since $X=0$ and $Y=0$, $v$ cannot reach facilities through its parent nor children, so it is disconnected and its weight $w(v)$ is included in the solution value.
We consider all the possible partitions of the budget $b$ into budgets $b_1,\ldots,b_k$ for the subtrees rooted at the children $u_1,\dots,u_k$. To find the best budget partition for all $0 \leq b \leq r$, we solve a multiple choice knapsack problem (MCKP). 
This knapsack variant requires exactly one item to be selected from every ``bin,'' and it can be solved by DP methods described in \cite[Section 11.5]{Kellerer04}.
For each budget partition, there are two ways to spend budget~$b_i$ on the subtree of child~$u_i$ while enforcing that no facilities are reachable in the subtree rooted at $v$ for guarantee $Y=0$:
\begin{itemize}
    \item Keep the edge $(v,u_i)$ and require that $v$ cannot reach any facilities through $u_i$, thus adding solution value $V_{00}(u_i,b_i)$ from the subtree rooted at $u_i$.
    
    \item Spend one unit of budget to remove edge $(v,u_i)$ and disconnect $v$ from the subtree rooted at~$u_i$; in this case, facilities in the subtree rooted at $u_i$ are not reachable from $v$, so we can maximize over values $V_{00}(u_i,b_i-1)$ and $V_{01}(u_i,b_i-1)$. However, $V_{00}(u_i,b_i-1) \leq V_{00}(u_i,b_i)$, so the former case does not improve the solution value and can be omitted. Intuitively, removing  $(v,u_i)$ does not improve the objective when neither $v$ nor $u_i$ can reach facilities.\qedhere
\end{itemize}
\end{proof}

\subparagraph*{Inner Nodes: Case $V_{10}$}

Similarly to the previous case, if $v$ is a facility, then it trivially reaches a facility in the subtree rooted at $v$, so $V_{10}(v,b) = -\infty$ for all $b\leq r$.

If $v$ is a customer, then we claim for all $b\leq r$, 
\begin{align}\label{eq:V10}
V_{10}(v,b) = \max_{b_i:\sum_{i=1}^k b_i = b} \Bigg\{ \;\sum_{i=1}^k\;
\max\bigg\{
V_{10}&(u_i,b_i),
V_{00}(u_i,b_i-1),
V_{01}(u_i,b_i-1)
\bigg\} \Bigg\}\,.
\end{align}

\begin{proof}
By assumption $X=1$, we have that $v$ can reach a facility through its parent and so its weight $w(v)$ is not included in the solution value.
Similarly to the previous case, we solve an instance of MCKP to obtain, for all budgets $b \leq r$, the optimal partitions of $b$ into budgets $b_1,\ldots,b_k$ for the subtrees rooted at the children $u_1,\dots,u_k$.
For each budget partition, we have three ways to spend the budget $b_i$ on the subtree of child $u_i$ while enforcing $Y=0$:
\begin{itemize}
    \item Keep the edge $(v,u_i)$ and require that $v$ cannot reach any facilities through $u_i$, thus adding solution value $V_{10}(u_i,b_i)$ from the subtree rooted in $u_i$.
    \item Spend one unit of budget to remove edge $(v,u_i)$ and disconnect $v$ from the subtree rooted at $u_i$; in this case, facilities in the subtree rooted at $u_i$ are not reachable from $v$, so we maximize over values $V_{00}(u_i,b_i-1)$ and $V_{01}(u_i,b_i-1)$.\qedhere
\end{itemize}
\end{proof}

\subparagraph*{Inner Nodes: Case $V_{X1}$}

With $Y=1$ guaranteeing that $v$ can reach a facility in its subtree, we claim that $V_{01}(v, b)=V_{11}(v, b)$, leading to a single set of recursive equations for $V_{X1}$.
\begin{lemma} \label{lem:VX1}
$V_{01}(v, b)=V_{11}(v, b)$ for all $v\in V$ and $b\leq r$.
\end{lemma}

\begin{proof}
Fix $v,b$. If $v$ is a facility, then child $u_i$ can reach a facility through $v$ if and only if edge $(v,u_i)$ is not removed. Thus, assumption $X$ on $v$'s ability to reach a facility through its parent does not affect whether any $u_i$ or, in fact, any other node in $v$'s subtree can reach a facility. Since $v$ has no weight to contribute, the value of $V_{X1}(v,b)$ is defined solely by the DP values of its children, and hence $V_{01}(v, b)=V_{11}(v, b)$.

If $v$ is a customer, then $Y=1$ guarantees that $v$ can reach a facility through one of its children, regardless of assumption $X$. Thus, each other child $u_i$ can reach a facility through $v$ (namely, in a sibling subtree) if and only if edge $(v,u_i)$ is not removed. Once more, $X$ does not affect whether other nodes in $v$'s subtree can reach a facility, giving $V_{01}(v, b)=V_{11}(v, b)$.
\end{proof}
We now give the recursive equations for $V_{X1}(v,b)$ when $v$ is a facility. Since $Y=1$ is trivially satisfied, for each child $u_i$, we just need to decide its budget~$b_i$ and whether to spend one unit of budget to disconnect $u_i$ from $v$. Thus, for all $b\leq r$,
\begin{align*}
V_{X1}(v,b) = \max_{b_i:\sum_{i=1}^k b_i = b} \Bigg\{ \;\sum_{i=1}^k\;
&\max\bigg\{
V_{10}(u_i,b_i),
V_{11}(u_i,b_i),
V_{00}(u_i,b_i-1),
V_{01}(u_i,b_i-1)
\bigg\} \Bigg\}\,.
\end{align*}
Since $V_{01}(u_i,b_i-1)=V_{11}(u_i,b_i-1)\leq V_{11}(u_i,b_i)$, the last term can be ignored, obtaining
\begin{align}\label{eq:V11_facility}
V_{X1}(v,b) = \max_{b_i:\sum_{i=1}^k b_i = b} \Bigg\{ \;\sum_{i=1}^k\;
\max\bigg\{
V_{10}&(u_i,b_i),
V_{11}(u_i,b_i),
V_{00}(u_i,b_i-1)
\bigg\} \Bigg\}\, .
\end{align}
Before giving the recursive equations for $V_{X1}(v,b)$ when $v$ is a customer, we note that we \emph{require} $v$ to be connected to a facility in its subtree, causing us to introduce a new variable $z_i$ indicating whether $v$ is connected to a facility through $u_i$. Thus, in addition to maximizing over budget partitions subject to $\sum_{i=1}^{k} b_i=b$, we also maximize over $z_1,\dots,z_k\subseteq \{0,1\}^k$ subject to the constraint that $\sum_{i=1}^{k} z_i\geq 1$. Now, when $v$ is a customer, we claim that for all $b\leq r$,
\begin{align}\label{eq:V11_customer}
V_{X1}&(v,b) = \max_{b_i,z_i:\sum_{i=1}^k b_i = b, \sum_{i=1}^k z_i \geq 1}
\Bigg\{ \;\sum_{i=1}^k\; z_i V_{X1}(u_i,b_i) \quad+ \notag\\
& (1-z_i) \max\left\{ V_{10}(u_i,b_i),V_{00}(u_i,b_i-1),V_{01}(u_i,b_i-1) \right\} \Bigg\}\,.
\end{align}
\begin{proof}
By assumption $Y=1$, we have that $v$ can reach a facility in its subtree, so weight $w(v)$ is not included in the solution value. We then solve a multiple choice knapsack problem across partitions of budget $b$ into budgets $b_1,\ldots,b_k$ for all $b\leq r$ and indicators $z_1,\dots,z_k$ such that $\sum_{i=1}^{k} z_i\geq 1$. For the subtree of child $u_i$, there are various ways to spend the budget $b_i$ depending on the value of $z_i$:
    \begin{itemize}
        \item When $z_i = 1$, we keep the edge $(v,u_i)$ and require that $v$ can reach a facility through $u_i$, adding solution value $V_{X1}(u_i,b_i)$ (as $V_{11}(u_i,b_i)=V_{01}(u_i,b_i)$ due to \cref{lem:VX1}).
        \item When $z_i=0$, we can keep the edge $(v,u_i)$ in the case that $u_i$ cannot reach a facility in its subtree, adding solution value $V_{10}(u_i,b_i)$. Otherwise, we can spend one unit of budget to remove the edge and add solution value $V_{01}(u_i,b_i-1)$ or $V_{00}(u_i,b_i-1)$.
    \end{itemize}
The additional constraint $\sum_{i=1}^{k} z_i \geq 1$ leads to a novel variant of the multiple choice knapsack problem, which we call the \emph{constrained multiple choice knapsack problem} (CMCKP). This can be solved by similar DP methods to MCKP, which we describe in \cref{alg:MCKP} in \cref{sec:MCKP}.
\end{proof}

\subsection{Time Complexity and Solution Reconstruction}\label{sec:reconstruction}

At every inner node $v$, our algorithm takes time $O(kr^2)$ to solve a (constrained) multiple-choice knapsack problem, where $k$ is the number of children of~$v$. To solve MCKP instances, we use a simple DP algorithm given in \cite[Section 11.5]{Kellerer04}, which we modify in \cref{sec:MCKP} to give \cref{alg:MCKP} for the constrained variant. The time complexity of our algorithm can then be obtained through a charging argument. At every inner node $v$, we will charge cost~$O(r^2)$ to each of the $k$ edges connecting $v$ to its children, which total to $O(kr^2)$ as desired.
Across all inner nodes, each edge is charged exactly once with cost $O(r^2)$; thus, the cost across inner nodes is $O(|E|r^2)=O(nr^2)$, dominating the $O(nr)$-time to initialize $V_{XY}(v,b)$ for all leaves $v$ and budgets $b\leq r$.
Hence, \myAlgorithm runs in $O(n r^2)$ time.


To explicitly reconstruct the set of removed edges $R^*$ that yield an optimal solution, our algorithm maintains a predecessor array $P_{XY}(v,b)$ of pointers from every \emph{state} of inner node $v$ (i.e., combination of conditions $XY$ and budget $b$) to a state for each child $u_i$ (i.e., combination of conditions $X_iY_i$ and budget $b_i$). These predecessors are set by taking an argmax at every max in our DP equations and storing the states in $P$. We also store indicator array $I$, which tracks whether edge $(v,u_i)$ was removed for each state of~$v$ and child $u_i$, equivalent to a DP equation allotting budget $b_i-1$ to the subtree of $u_i$.

After calculating the optimal value $V_{X^*Y^*}(v^*,r)$ at root $v^*$ from $XY\in\{00,01,10,11\}$, we traverse the tree downwards by following our predecessor array $P$. At inner node~$v$ in state~$(XY,b)$, $P_{XY}(v,b)$ gives pointers to states for each child $u_i$ and indicator array $I_{XY}(v,b)$ tells us whether $(v,u_i)\in R^*$. Solution reconstruction takes time $O(n)$ for a post-order traversal of the tree, which is dominated by the runtime of solving the DP. As for memory, both solving the DP and solution reconstruction require only $O(nr)$ space.

\section{Numerical Results for \myAlgorithm}
\label{sec:experimental}
We evaluate \myAlgorithm on a data set of randomly generated problem instances to highlight the characteristics of both our fixed-parameter linear solution and the $r$-edge interdiction covering problem. All numerical results were obtained on an Apple M1 CPU.

To determine appropriate baselines for our numerical evaluation, we observe (similarly to Fröhlich and Ruzika~\cite{FrohlichR22}) that REIC on trees can be formulated as the following integer program, where disconnecting a customer $v \in V\setminus S$ (i.e., setting $x_v=0$) requires that at least one edge $e$ is removed (setting $y_e=1$) on every path from $v$ to a facility $s \in S$. We let $P_{uv}$ denote the unique $u$-$v$ path in $T$.

\begin{equation}\label{eq:ip}\tag{REIC-IP}
\begin{array}{rr@{}llll}
\text{max}&\displaystyle\sum_{v\in V\setminus S} w(v)\,(1-x_v) \\
\text{subject to}& x_v + \displaystyle\sum_{e\in P_{sv}} y_e \ &\geq 1 \quad&\forall v\in V\setminus S, s\in S\\
    &\displaystyle\sum_{e\in E} y_e \ &\leq r\\
    &  x_v \ &\in \{0,1\} \quad&\forall v\in V\setminus S \\
    & y_e \ &\in \{0,1\}\quad&\forall e\in E
\end{array}
\end{equation}

\begin{figure*}[htb]
  \centering
  \setlength{\tabcolsep}{6pt}
  \begin{tabular}{cc}
    \subfloat[Runtime as a function of the number of nodes $n$
      (for $p = 0.4$ and $r = n/10$)\label{subfig:nodes_CI}]{%
      \includegraphics[width=0.405\textwidth]{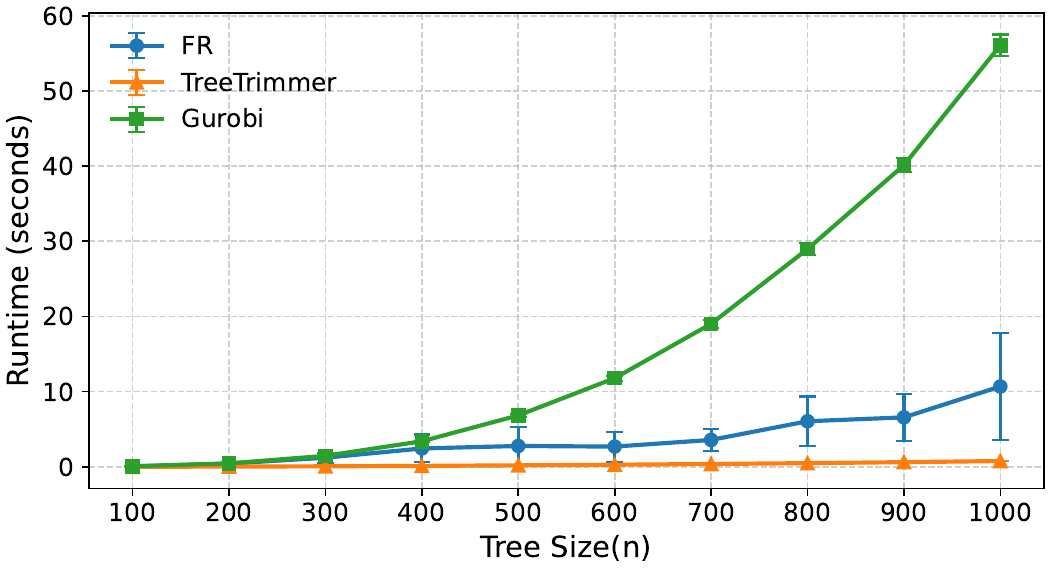}} &
    \subfloat[Speedup as a function of the number of nodes $n$
      (for $p = 0.4$ and $r = n/10$)\label{subfig:nodes_speed}]{%
      \includegraphics[width=0.40\textwidth]{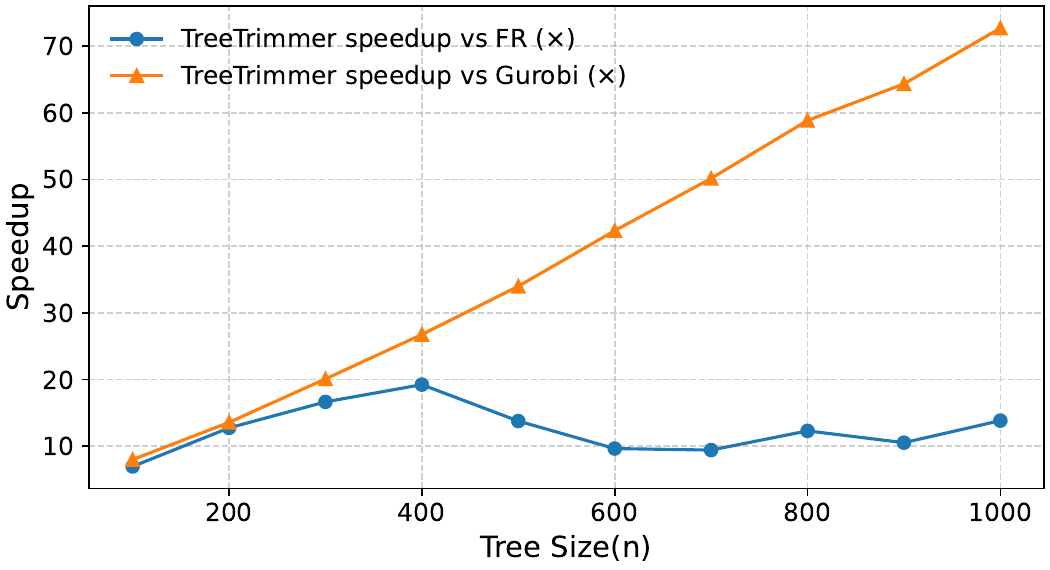}} \\

    \subfloat[Runtime as a function of the facility probability $p$
      (for $n = 200$ and $r = n/10$)\label{subfig:facility_CI}]{%
      \includegraphics[width=0.42\textwidth]{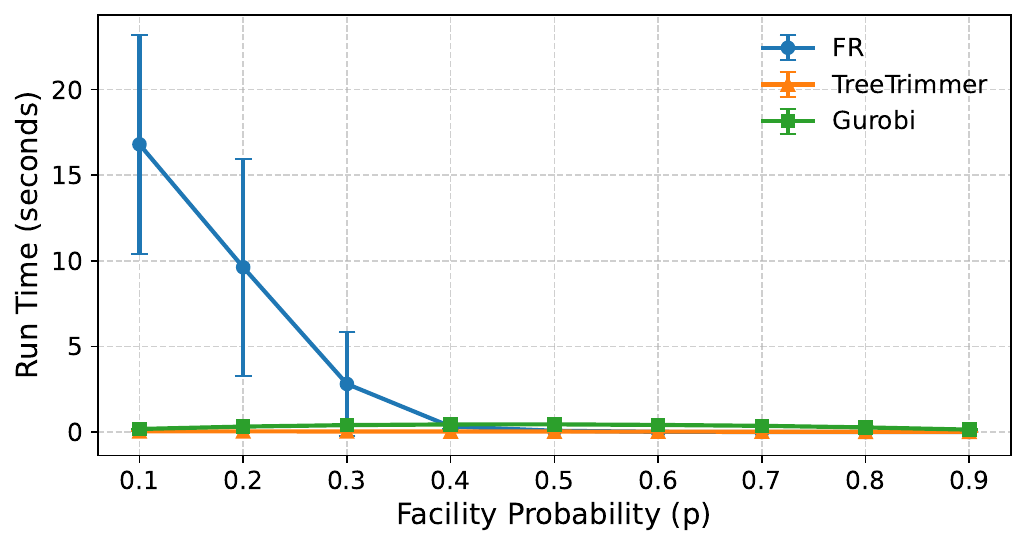}} &
    \subfloat[Runtime as a function of the budget $r$
      (for $n = 600$ and $p = 0.4$)\label{subfig:budget_CI}]{%
      \includegraphics[width=0.405\textwidth]{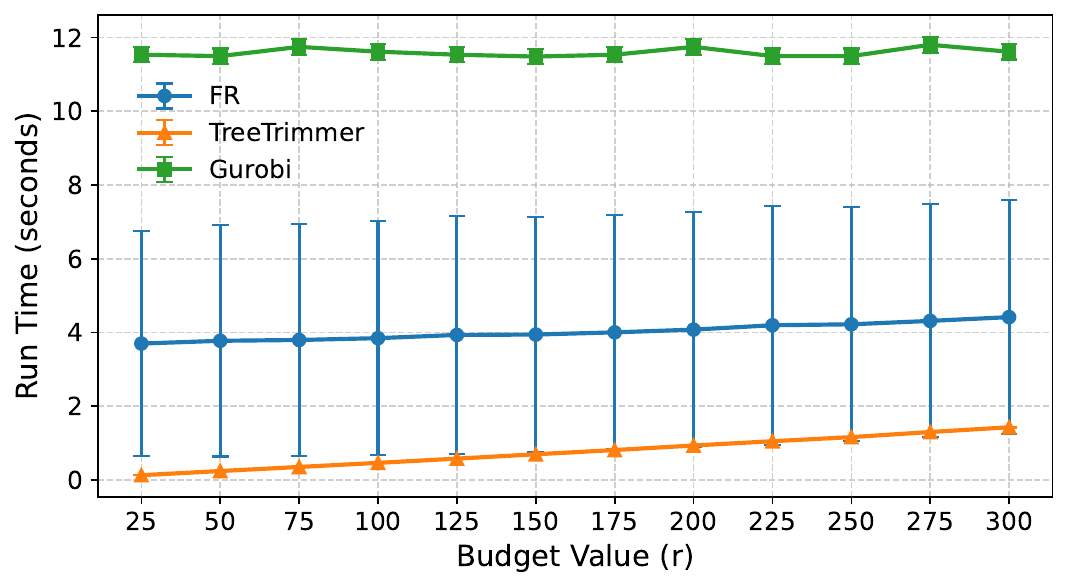}} \\
  \end{tabular}
  \caption{Comparison of average mean runtime under varying
           experimental settings.}
  \label{fig:overall_results}
\end{figure*}
\subsection{Baselines and Dataset}
We use two baselines for solving REIC on trees: (1) Gurobi Optimizer~\cite{Gurobi} (using the formulation of \ref{eq:ip}), and (2) Fröhlich and Ruzika method \cite{FrohlichR22}.
For our dataset, we generate random tree instances using the \texttt{nx.random\_tree(n, seed)} function from the NetworkX library, which produces labeled trees with exactly $n$ nodes \cite{hagberg2008exploring}, uniformly at random; this function generates a random Prüfer sequence of length $n - 2$ and converts it to a tree, ensuring that all possible labeled trees are equally likely. Setting a fixed seed guarantees reproducibility across experimental runs. We generate two types of datasets: (i) where any node in the network can serve as a facility, and (ii)
where only leaf nodes can serve as facilities. This distinction allows direct comparison with the solution of subproblems in \cite{FrohlichR22} where facilities are located on leaf nodes.

\subsection{Experiments}
To illustrate \myAlgorithm's performance, we vary the parameters $n$ and $r$ as well as $p$, the probability of a node being a facility.
To obtain statistically meaningful results, we run each problem instance on a set of random trees and report averages and corresponding confidence intervals. Our experiments focus on: (i) comparing performance of \myAlgorithm  with above baselines (in \cref{subsubsection:comparison}), and (ii) analyzing sensitivity of runtime performance to changes in above parameters (in \cref{subsubsection:sensitivity}). 
We  illustrate that the performance gap between \myAlgorithm  and FR as well as Gurobi grows larger as we increase $n$ or decrease $p$. This gap is also larger with lower $r$ values. 
We show that FR's behavior can be non-monotonic as a function of these parameters, and it tends to exhibit higher runtime variability.
Thus, we also (iii) conduct experiments to provide insight into sources of FR's high variance (in \cref{subsubsection:source}),
showing that it is largely a function of differences in cluster sizes, shape, and the spatial distribution of facilities; when the number of facilities is small, FR forms fewer but larger clusters, increasing both runtime mean and variance.
We present representative results but conducted experiments across a broad range of parameter combinations, with qualitatively similar results.

\subparagraph*{Runtime Performance}
\label{subsubsection:comparison}
We now compare the mean runtime of \myAlgorithm, FR, and Gurobi.

In \cref{subfig:nodes_CI}, we depict the mean runtime as a function of $n \in \{100, 200, \ldots, 1000\}$ for $p = 0.4$ and $r = n/10$ and plot the corresponding 95\% confidence intervals. \cref{subfig:nodes_speed} depicts  corresponding speedup achieved by \myAlgorithm, relative to baseline methods; \myAlgorithm achieves a lower mean runtime than Gurobi and FR, with speedup of $6.9$--$19.2\times$ over FR and $8$--$72.6\times$ over Gurobi. In Gurobi's case, larger $n$ results in greater speedup, since (as depicted in \cref{subfig:nodes_CI}) Gurobi's runtime grows exponentially with $n$. In contrast, the speedup over FR is not monotonic as a function of $n$; we believe this to be due to FR's higher runtime variability, as can be seen by the relatively large confidence interval range in \cref{subfig:nodes_CI} (we explore this in greater detail in \cref{subsubsection:sensitivity,subsubsection:source}).

\begin{figure*}[tbh]
  \centering
  \setlength{\tabcolsep}{6pt}
  \begin{tabular}{ccc}
    \subfloat[CV as a function of varying number of nodes.\label{subfig:nodes_CV}]{%
      \includegraphics[width=0.30\textwidth]{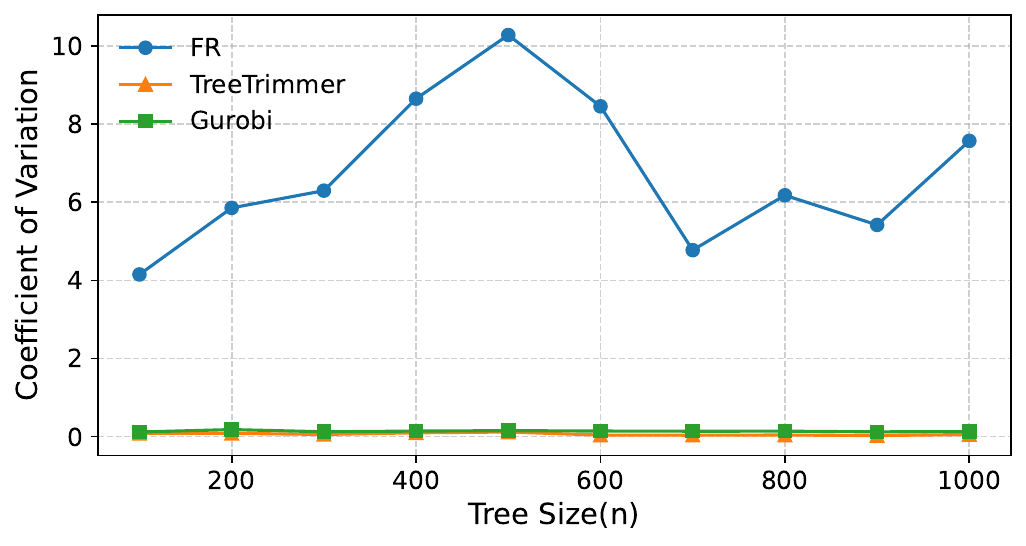}} &
    \subfloat[CV as a function of varying facility probability.\label{subfig:facility_CV}]{%
      \includegraphics[width=0.30\textwidth]{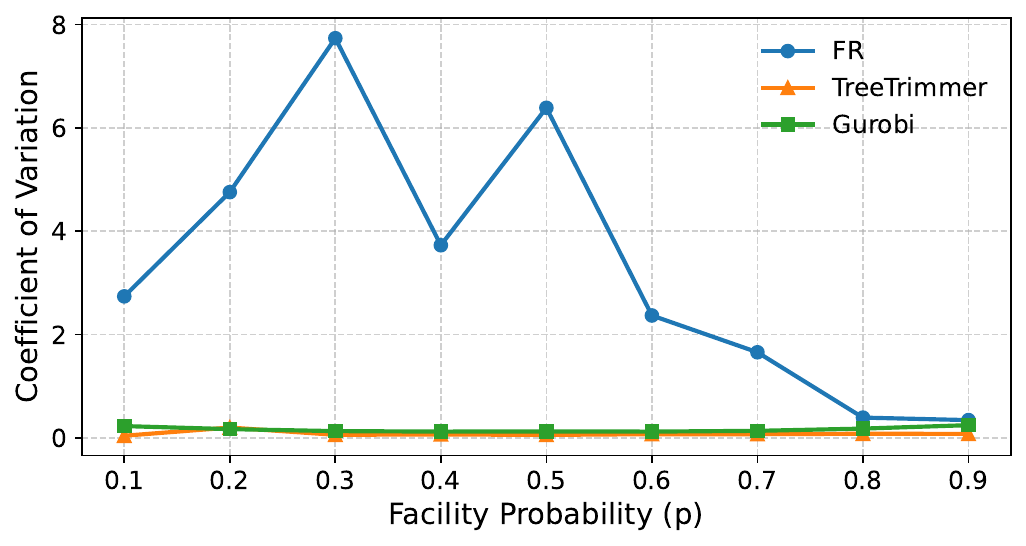}} &
    \subfloat[CV as a function of varying budget value.\label{subfig:budget_CV}]{%
      \includegraphics[width=0.30\textwidth]{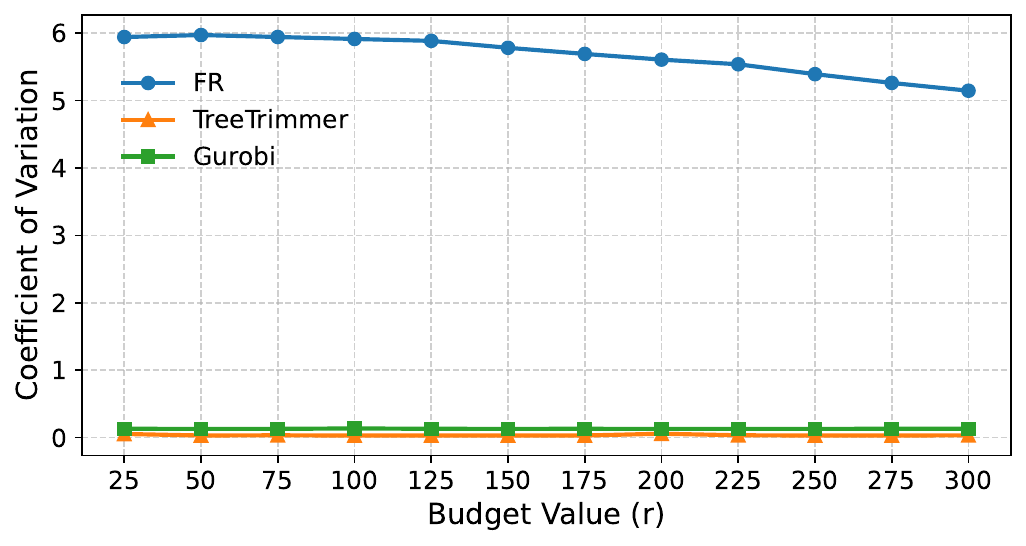}} \\
  \end{tabular}
  \caption{Comparison of sensitivity to different variables:
  (a) coefficient of variation as a function of total number of nodes (tree size),
  (b) coefficient of variation as a function of facility probability,
  (c) coefficient of variation as a function of budget values.}
  \label{fig:overall_results_CV}
\end{figure*}

\subparagraph*{Runtime vs.\ $p$}

\cref{subfig:facility_CI} depicts mean runtime (with corresponding 95\% confidence interval) as a function of $p \in \{0.1, 0.2, \ldots, 0.9\}$ for $n = 200$ and $r = n/10$. Decreasing $p$ affects the three algorithms in different ways. Lowering $p$ substantially increases FR's mean runtime. This occurs because FR forms subproblems as clusters of customer nodes where facilities are located on leaves of the original tree: When $p$ decreases, it is more likely that these clusters grow, in turn increasing runtime mean and variance; we explore the underlying reason for this behavior in \cref{subsubsection:source}. For Gurobi, the most computationally challenging instances occur when $p = 0.5$, as this setting produces the largest number of ILP constraints, with runtime approximately $3\times$ the fastest setting (i.e., with very low or very high values of $p$). For \myAlgorithm, the most demanding cases occur when $p$ is low; however, it is still significantly faster, e.g., approximately $420\times$ faster than FR and $4.5\times$ faster than Gurobi. In this setting, the smaller number of facilities slightly increases the amount of computation required for customer nodes in \myAlgorithm. However, this effect remains relatively minor compared to the strong sensitivity observed in FR and Gurobi. 

\subparagraph*{Runtime vs.\ $r$}

\cref{subfig:budget_CI} depicts the mean runtime (and corresponding 95\% confidence interval) as a function of $r \in \{25, 50, \ldots, 300\}$ for $p = 0.4$ and $n = 600$. \myAlgorithm  is most sensitive to $r$, although its runtime remains lower. Although \myAlgorithm's theoretical complexity is $O(nr^2)$, the observed runtime trend is closer to linear as a function of $r$; we believe that this is due to the constant factor associated with enumerating possible items, $O(nr)$, being larger than that of the MCKP constraint computation, $O(nr^2)$, in our Python implementation. FR exhibits a sublinear trend as a function of $r$, we believe, due to its clustering approach where for each cluster, the runtime tends to saturate once $r$ becomes comparable to the number of facilities in that cluster. Gurobi’s runtime is not significantly affected by increases in $r$, since the number of constraints in its formulation remains constant. FR continues to exhibit relatively large confidence intervals range in \cref{subfig:budget_CI} (explored further in \cref{subsubsection:sensitivity} and \cref{subsubsection:source}).

\begin{table}[h]
\centering
\caption{Mean runtime (seconds) with 95\% confidence intervals (mean $\pm$ CI) across a range of $p$ values, for $n = 600$ and $r = n/10$.}
\label{tab:runtime_ci}
\begin{tabular}{|c|c|c|c|}
\hline
$p$ & \myAlgorithm\ (s) & Gurobi (s) & FR (s) \\
\hline
0.3 & 0.310 $\pm$ 0.002 & 10.236 $\pm$ 0.192 & 28.842 $\pm$ 36.867 \\
0.4 & 0.284 $\pm$ 0.003 & 11.680 $\pm$ 0.210 & 3.561 $\pm$ 2.538   \\
0.5 & 0.254 $\pm$ 0.005 & 11.859 $\pm$ 0.220 & 0.622 $\pm$ 0.804   \\
0.6 & 0.224 $\pm$ 0.002 & 11.152 $\pm$ 0.214 & 0.117 $\pm$ 0.062   \\
0.7 & 0.194 $\pm$ 0.002 & 9.397 $\pm$ 0.189  & 0.054 $\pm$ 0.015   \\
0.8 & 0.166 $\pm$ 0.002 & 6.887 $\pm$ 0.149  & 0.031 $\pm$ 0.001   \\
0.9 & 0.134 $\pm$ 0.001 & 3.640 $\pm$ 0.096  & 0.017 $\pm$ 0.000   \\
\hline
\end{tabular}
\end{table}

Lastly, we repeat the experiment from \cref{subfig:facility_CI} but with $n=600$ to better understand the runtime gap among the three algorithms as the number of nodes grows. As shown in \cref{tab:runtime_ci}, FR's runtime for smaller values of $p$ grows substantially when $n$ increases from $200$ (in \cref{subfig:facility_CI}) to $600$ (in \cref{tab:runtime_ci}), whereas the runtime of \myAlgorithm\ remains quite small. Gurobi's runtime is not small but with low variance. FR's runtime is large and highly variable (relative to the mean) for more reasonable values of $p=0.3, 0.4$. Note that, FR's mean runtime at $p=0.3$ is around $8\times$ higher than at $p=0.4$. And, if we compare the runtime at $p=0.3$ with that in \cref{subfig:facility_CI}, we observe an order of magnitude growth (in both mean and confidence interval), from $2.810 \pm 3.012$ (for $n=200$) to $28.842 \pm 36.867$ (for $n=600$). FR's mean runtime and its variability are small when $p$ reaches high values ($0.6$ and above in \cref{tab:runtime_ci}), which is not a reasonable setting for most applications (i.e., to have more than half of the nodes be facilities).
We omit results for $p \in \{0.1, 0.2\}$ because FR's runtime grows so large over this parameter range that we were not able to complete our experiments for $n=600$ (these results are available in \cref{subfig:facility_CI} for $n=200$).
We also note that a typical application, like a LEO constellation, would have thousands of nodes (e.g., around 9,000 in Starlink).

\subsection{Sensitivity of Runtime Performance}
\label{subsubsection:sensitivity}
We now repeat above experiments but instead consider the coefficient of variation (CV) as our metric, again as a function of $n$, $r$, and $p$, in order to explore the variability in runtime behavior of \myAlgorithm, FR, and Gurobi. The results are depicted in \cref{fig:overall_results_CV}.

As $n$ increases, CV remains close to zero and nearly constant for both Gurobi and \myAlgorithm; in contrast, FR's CV is much higher (ranging between approximately 4 and 10) and non-monotonic, as shown in \cref{subfig:nodes_CV}. 
As $r$ increases, CV again remains close to zero and nearly constant for Gurobi and \myAlgorithm. FR's CV is substantially higher but does not vary as much as a function of $r$, as shown in \cref{subfig:budget_CV}.

As $p$ increases, CV remains close to zero and nearly constant for Gurobi and \myAlgorithm; in contrast, FR's CV is much higher (ranging between approximately 0.3 and 7.7) and non-monotonic, as shown in
\cref{subfig:facility_CV}.

Thus, we observe a low CV for Gurobi and \myAlgorithm. In contrast, FR can experience high CV; moreover, its behavior is non-monotonic with respect to $n$ and $p$. We explore underlying source of the high variation and non-monotonic behavior of FR next.

\begin{figure*}[t]
  \centering
  \setlength{\tabcolsep}{6pt}
  \begin{tabular}{cc}
    \subfloat[Runtime as a function of nodes for one cluster\label{subfig:FR_box}]{%
      \includegraphics[width=0.40\textwidth]{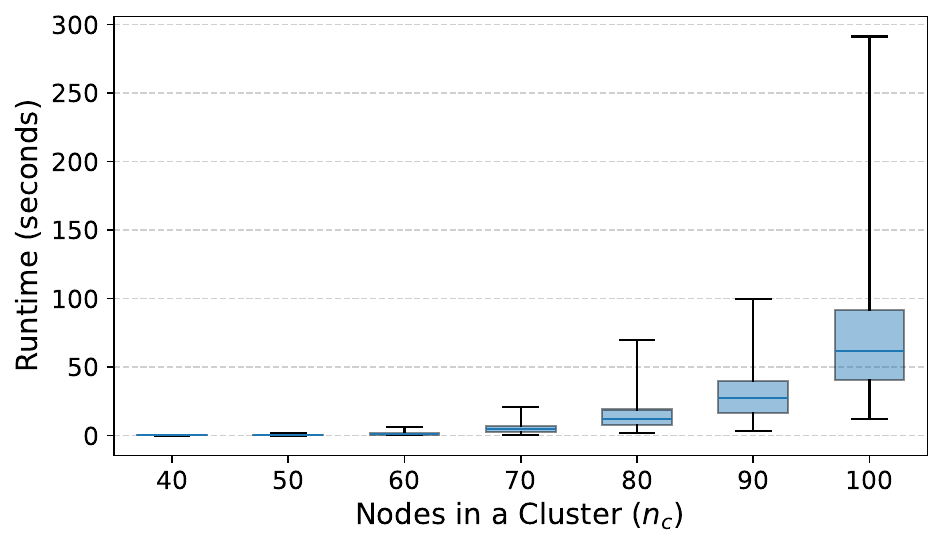}} &
    \subfloat[Mean runtime as a function of nodes for one cluster\label{subfig:FR_CI}]{%
      \includegraphics[width=0.43\textwidth]{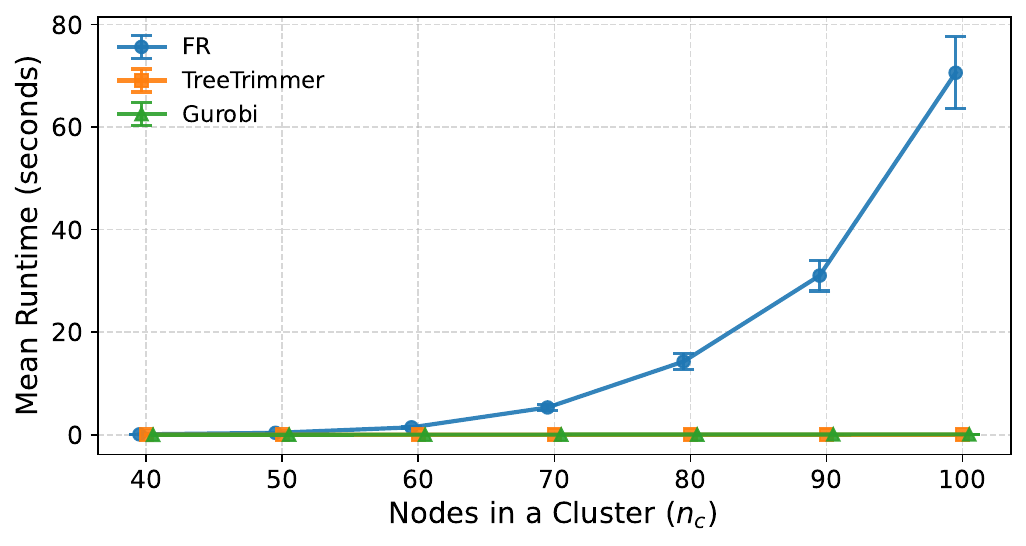}} \\
    \subfloat[Runtime as function of customer joints\label{subfig:FR_JC_box}]{%
      \includegraphics[width=0.43\textwidth]{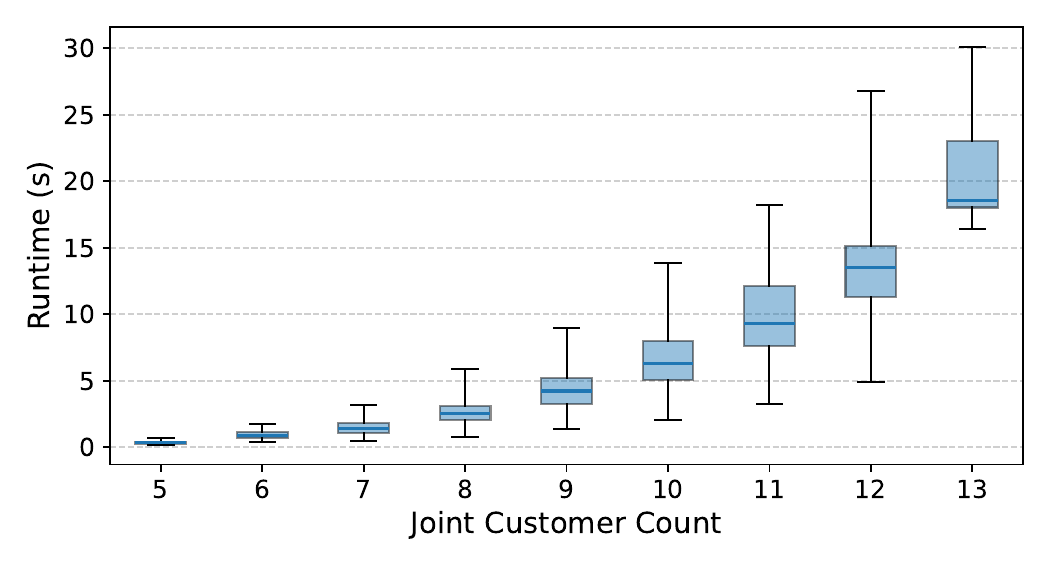}} &
    \subfloat[CV as a function of customer joints\label{subfig:FR_JC_CV}]{%
      \includegraphics[width=0.43\textwidth]{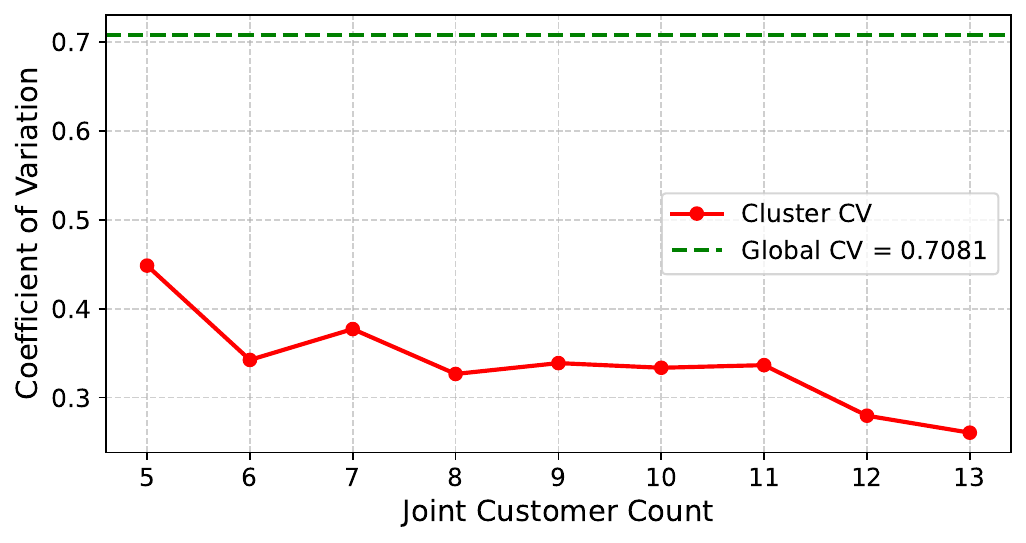}} \\
  \end{tabular}
  \caption{Illustration of the source of high variance in FR runtime:
  (a) runtime box plot for different numbers of nodes in the cluster, and
  (b) mean runtime as a function of the number of nodes in the cluster,
  (c) runtime box plot for varying number of customer joints in the cluster,
  showing that an increase in the number of customer joints increases the
  runtime of FR, and (d) CV as a function of customer joints.}
  \label{fig:overall_results_FR}
\end{figure*}

\subsection{Insight into FR's High Performance Variability}
\label{subsubsection:source}

We now conduct experiments to provide insight into sources of FR's high variance.

\subparagraph*{Cluster Size}

FR first identifies clusters of customer nodes with facilities located on leaves and then computes the results for each cluster across all budget values. Thus, to understand FR's behavior, we consider experiments with single cluster trees. That is, we compute the runtime of each algorithm over 200 randomly generated instances of trees with a single cluster, i.e., in these instances facilities are restricted to leaves. We do this for $n_c \in \{40, 50, \ldots, 100\}$ nodes in a cluster, with $r = n_{c}/10$ and $p = 0.5$. We first depict the mean runtime as a function of $n_c$ with corresponding 95\% confidence interval in \cref{subfig:FR_CI}. As expected, the mean runtime of the FR algorithm grows much faster than that of our algorithm and Gurobi since FR's worst-case complexity is $O(rn^7)$ (we can consider $n=n_c$ for a single cluster). We observe that even within a single cluster, the width of the confidence interval is significant for larger clusters. To explore the variance behavior, we depict the corresponding box plot of runtime, indicating the minimum, 25th percentile (Q1), median, 75th percentile (Q3), and maximum runtime in \cref{subfig:FR_box}. Note that for a cluster of a particular size, the variation in runtime can be quite high, which occurs because the structure of the cluster significantly affects the runtime, i.e., two clusters with the same number of nodes (customers and facilities) may still exhibit significantly different runtimes depending on their internal structure, and the larger the cluster, the greater is the chance of outliers in its structure. This explains the results in \cref{subfig:facility_CI}, where low $p$ produces relatively large clusters and high variability in runtime.

\subparagraph*{Customer Joint Nodes}

We define customer joints as customer nodes whose removal results in at least three connected components, each containing a facility. In this experiment, we compute the mean runtime across 2,000 randomly generated single cluster trees, i.e., with facilities restricted to leaves, with $p = 0.5$, $n_c = 70$, and $r = 7$. For each instance, we record the number of customer joints and the corresponding runtime. We then group instances based on their joint counts and generate box plots summarizing the minimum, first quartile (Q1), median, third quartile (Q3), and maximum runtime for each group, as depicted in \cref{subfig:FR_JC_box}. 
We also report the CV of runtime across groups, as shown in \cref{subfig:FR_JC_CV}, where the line labeled as ``global CV'' indicates the CV computed over all 2000 instances, depicted for comparison. We observe that the runtime increases with the number of customer joints. This occurs because edges incident to customer joint nodes disconnect groups of nodes from at least two facilities when removed. Such edges are looked at frequently in FR, increasing the computational cost. We also observe in \cref{subfig:FR_JC_CV} that the CV is low when viewed as a function of the number of customer joints (and lower than the ``global CV''); thus graphs with similar numbers of customer joints tend to exhibit similar runtimes.

\subparagraph*{Summary}
The runtime variation of FR is largely a function of differences in cluster sizes, shape, and the spatial distribution of facilities. When the number of facilities is small, FR forms fewer but larger clusters, increasing both mean runtime and its variance. This behavior is further influenced by the presence of customer joints: as their number increases, the median runtime grows.

\section{\texorpdfstring{$r$}{r}-Edge Interdiction Covering on Bounded Treewidth Graphs}
\label{sec:bddtw}

\subsection{Tree Decompositions and Treewidth}

We also study the $r$-edge interdiction covering problem on the class of bounded treewidth graphs. Informally, treewidth is a measure of how far a graph is from being a tree, with trees and forests having treewidth 1. We now give definitions of tree decompositions and treewidth, following the treatment of \cite[Chapter 7]{Cygan2015ParameterizedA}.
\begin{definition}
    A \emph{tree decomposition} $(T,\{X_t\}_{t\in V(T)})$ of graph $G$ is a tree $T$ in which each node $t$ is associated with a set of vertices $X_t\subseteq V(G)$, called a bag, such that the following conditions hold:
    \begin{itemize}
        \item For every vertex $v\in V(G)$, there is a node $t\in T$ such that $v\in X_t$.
        \item For every edge $(u,v)\in E(G)$, there is a node $t\in T$ such that $X_t$ contains $u$ and $v$.
        \item For every vertex $v\in V(G)$, the nodes whose corresponding bags contain $v$ form a connected subtree of $T$, i.e., $T_v=\{t\in V(T):v\in X_t\}$ induces a connected subtree on $T$.
    \end{itemize}
The \emph{width} of a tree decomposition $(T,\{X_t\}_{t\in V(T)})$ is defined to be $\max_{t\in T}|X_t|-1$, the size of the largest bag minus one. The \emph{treewidth} of a graph $G$ is the minimum width over all tree decompositions of $G$, denoted as $\tw(G)$. 
\end{definition}
Throughout this section, we will use the term vertices to denote the elements of $V(G)$ and nodes for the elements of $V(T)$ in order to avoid confusion.

Generally, it is NP-complete to decide whether a graph has treewidth at most $k$, for arbitrary $k$. Bodlaender \cite{bodlaender} gave an FPT algorithm to recognize graphs of treewidth at most $k$ and construct a width $k$ tree decomposition. This algorithm is linear time for constant $k$, i.e., for graphs of ``bounded treewidth,'' but with a runtime of $k^{O(k^3)}\cdot n$, it is highly impractical for any real setting. On the other hand, if one settles for approximate rather than exact treewidth, then there is a $4$-approximation with single-exponential runtime of $O(8^kk^2\cdot n^2)$ due to \cite{RobertsonSeymourTW} as well as a polynomial-time $O(\sqrt{\log k})$-approximation \cite{FHLTW}.

In order to describe a dynamic programming algorithm on tree decompositions, it will be helpful to work with a convenient kind of tree decomposition called a nice tree decomposition, and in particular, an \emph{extended nice tree decomposition}. Originating in the work of Kloks~\cite{Kloks1994TreewidthCA}, these decompositions help to define DP equations for our interdiction problem as they introduce vertices and edges one by one.
\begin{definition}
    An \emph{extended nice tree decomposition} is a tree decomposition with fixed root node $q\in V(T)$ such that $X_q=\emptyset$, where every node in $T$ is one of the following:
    \begin{itemize}
        \item \emph{Leaf node:} a leaf $t$ of $T$ with $X_t=\emptyset$.
        \item \emph{Introduce vertex node:} a node $t$ with exactly one child $t'$ such that $X_t=X_{t'}\cup \{v\}$ for some $v\in V(G)\setminus X_{t'}$.
        \item \emph{Introduce edge node:} a node $t$ labeled by an edge $(u,v)\in E(G)$ with exactly one child $t'$ such that $X_t=X_{t'}$ and $u,v\in X_t$.
        \item \emph{Forget node:} a node $t$ with exactly one child $t'$ such that $X_t=X_{t'}\setminus \{w\}$ for some $w\in X_{t'}$.
        \item \emph{Join node:} a node $t$ with exactly two children $t_1,t_2$ such that $X_t=X_{t_1}=X_{t_2}$.
    \end{itemize}
Finally, we require that each edge of $E(G)$ is introduced exactly once in the tree decomposition.
\end{definition}
It is known that given a tree decomposition $(T,\{X_t\}_{t\in V(T)})$ of $G$ of width at most $k$, one can compute a width $k$ extended nice tree decomposition of $G$ with $O(k|V(G)|)$ nodes in time $O\left(k^3 |V(G)|\right)$ \cite{Cygan2015ParameterizedA}. Thus, for graphs of bounded treewidth, we can in linear time find a width $k$ extended nice tree decomposition of linear size.

\subsection{HedgeTrimmer: a DP Algorithm for REIC on Bounded Treewidth Graphs}
Let $(T,\{X_t\}_{t\in T})$ be a tree decomposition for $G$ with width $k=\tw(G)$, and let $G$ have $n$ vertices. We can assume without loss of generality that this is an extended nice tree decomposition. For any node $t$, $V_t$ is the union of all bags in the subtree rooted at $t$. A key property of a tree decomposition is that $X_t$ separates $V_t$ from $G\setminus V_t$; specifically, there are no edges in $G$ between $V_t\setminus X_t$ and $V(G)\setminus V_t$. Given that we are focused on connection, it will also be helpful to associate with node $t$ a subgraph $G_t$ of all edges introduced in its subtree:
\[
G_t=\left(V_t,E_t:=\{e:e\text{ is introduced in the subtree rooted at }t\}\right)
\]
Following our bottom-up dynamic programming approach for \myAlgorithm, we propose \ourAlgorithm  for solving REIC on bounded treewidth graphs; to find the solution on a subgraph $G_t$, our algorithm combines the solutions from each of the subgraphs of node $t$'s children. We extend the previous ideas of guarantees $X$ and $Y$ into a \emph{labeling} of a bag $X_t$. Our labeling is a mapping $f:X_t\rightarrow \{0,1\}$ which assigns a condition to each vertex in a bag:
\begin{itemize}
    \item $f(v)=0$: it is guaranteed that $v$ will \emph{not} be able to reach a facility;
    \item $f(v)=1$: it is assumed that $v$ will \emph{be connected} to a facility.
\end{itemize}
Thus, for a node $t$, there are $2^{|X_t|}$ possible labelings of bag $X_t$, as opposed to our four combinations of $XY$ conditions for trees. The key idea of our labeling is that we will only allow new edges to be added between vertices with like labelings, preventing $v$ with $f(v)=0$ from reaching a facility after being connected to $u$ with $f(u)=1$.

 We now define our DP variables for bounded treewidth graphs; for a node $t$, budget $b\leq r$, and a labeling $f$ of $X_t$, let $V(t,b,f)$ denote the maximum total weight of a set of customers $A\subseteq V_t\cap (V(G)\setminus S)$ such that, for a set of edges $R\subseteq E$:
\begin{itemize}
    \item $A\cap X_t=f^{-1}(0)$, i.e., vertices in bag $X_t$ are in $A$ if and only if labeled $0$;
    \item No vertex in $A$ is connected to a facility in $G_t\setminus R$;
    \item $|R|\leq b$.
\end{itemize}
Such a set $R$ is an interdiction strategy on $G$, and an optimal interdiction strategy with budget~$r$ has objective value exactly $V(q,r,f_\emptyset)$, where $q$ is the root of the tree decomposition and $f_\emptyset$ is the empty function. This is the solution value because subgraph $G_q=G$ is the entire graph and, as $X_q=\emptyset$ by definition, the only possible labeling $f:X_q\rightarrow \{0,1\}$ is $f_\emptyset$.

For a labeling $f$ of $X$, we will let $f|_Y$ denote the restriction of $f$ to $Y\subseteq X$. Similarly, we will let $f_{v\mapsto\alpha}$ denote the extension of $f$ where $v$ is added to the domain and mapped to $\alpha$:
\[
f_{v\mapsto \alpha}(u)=\begin{cases}
    f(u)&\text{ if }u\neq v,\\
    \alpha&\text{ if }u= v.
\end{cases}
\]

\subsubsection{DP Equations}
We present the DP equations for \ourAlgorithm, beginning with the base case of leaves and then giving our recursive cases. We note that for a facility $v$ we must have $f(v)=1$, so DP values where $f(v)=0$ are set to $-\infty$.

\subparagraph*{Leaf Node}
Let $t$ be a leaf node, giving that $X_t=\emptyset$. We set $V(t,b,f_\emptyset)=0$ for $b\geq 0$, as the only possible labeling is the empty labeling.

\subparagraph*{Introduce Vertex Node}
Let $t$ be an introduce vertex node with child node $t'$ such that $X_t=X_{t'}\cup \{v\}$ and $v\notin X_{t'}$. As vertices are introduced without any edges, customers introduced are inherently not connected to facilities. Thus, we add value $w(v)$ if and only if we make the guarantee that customer $v$ will not be connected to a facility in the future:
\[
V(t,b,f)=\begin{cases}
    V(t',b,f|_{X_{t'}}) &\text{ if }f(v)=1,\\
    w(v)+V(t',b,f|_{X_{t'}}) &\text{ if }f(v)=0 \text{ and $v$ is customer},\\
    -\infty &\text{ otherwise. }\\
\end{cases}
\]
\subparagraph*{Introduce Edge Node}
Let $t$ be an introduce edge node for edge $(u,v)\in E(G)$ with $u,v\in X_t$ and with child node $t'$ such that $X_t=X_{t'}$. If our labeling of $u$ matches our labeling $v$, then we need not remove the edge to maintain our guarantees (or lack thereof) of facility reachability. More specifically, either $f(u)=f(v)=1$ and both vertices are guaranteed to be disconnected from facilities, or $f(u)=f(v)=0$ and neither vertex have such guarantee; in both cases, $u$ and $v$ can safely be connected while maintaining our guarantees. However, if the labelings of these vertices disagree, say $f(u)=0$ and $f(v)=1$, then we must remove the edge $(u,v)$ so that $u$ cannot reach a facility through $v$, forcing us to spend a unit of budget:
\[
V(t,b,f)=\begin{cases}
    V(t',b,f) &\text{if }f(v)=f(u),\\
    V(t',b-1,f) &\text{otherwise.}
\end{cases}
\]
\subparagraph*{Forget Vertex Node}
Let $t$ be a forget vertex node with child node $t'$ such that $X_t=X_{t'}\setminus \{w\}$ and $w\in X_{t'}$. Since $w\notin X_t$, the new labeling $f$ is compatible with any extension of $f$, allowing $w$ to be mapped to either $0$ or $1$ for bag $X_{t'}$; we take the max over both to find the optimal value:
\[
V(t,b,f)=\max\left\{V(t',b,f_{w\mapsto 0}),V(t',b,f_{w\mapsto 1})\right\}\,.
\]

\subparagraph*{Join Node}
Let $t$ be a join node with two children nodes $t_1,t_2$ such that $X_t=X_{t_1}=X_{t_2}$. The interdiction budget must be split between both children, but we require that both labelings agree completely as a customer can only be disconnected from facilities if it is disconnected from facilities in both subtrees $G_{t_1}$ and $G_{t_2}$. We subtract the weight of customers in $f^{-1}(0)$ from our new solution value to avoid double-counting the value of guaranteed-disconnected customers in $X_t$, where $f^{-1}(0)$ denotes the preimage of $0$ in $f$. To see that customers $v\notin X_t$ can only be counted once, observe that the nodes whose bags contain $v$ would induce a disconnected subtree on $T$ if $v$ was present in bags in the subtrees of both $X_{t_1}$ and $X_{t_2}$. Thus,
\[
V(t,b,f)=\max_{0\leq b_1\leq b}\left\{V(t_1,b_1,f)+V(t_2,b-b_1,f)-\sum_{v\in f^{-1}(0)}w(v)\right\}\,.
\]

\subsubsection{Analysis}

\subparagraph*{Correctness} 
We now show that \ourAlgorithm  is both feasible and optimal. First, we note that any interdiction strategy $R$ defines a global labeling $f^R$ that is respected by some $f^t$ at every bag $X_t$, i.e., $f^R|_{X_t}=f^t$. The feasibility of our DP equations follows from the observation that preserving our labeling $f^R$, and thus the set of disconnected customers, is exactly equivalent to removing all edges between vertices $u$ and $v$ where $f^R(u)\neq f^R(v)$. Therefore, every DP value corresponds to some interdiction strategy.

For optimality, consider the optimal interdiction strategy $R^*$ with $|R^*|\leq r$, which fixes some global labeling $f^{R^*}$. Trivially, $V(q,f_\emptyset,r)$ respects this labeling. As our DP equations take the maximum at every node of tree decomposition $T$, then our algorithm will output a solution with value at least that of $R^*$.

\subparagraph*{Time Complexity and Solution Reconstruction}

Our tree decomposition $T$ has size $|V(T)|=O(kn)$. At every node $t$, we must calculate $2^{|X_{t'}|}\cdot r\leq 2^k\cdot r$ new values of $V$, one for every state (i.e., pair of labelings $f$ and budgets $b\leq r$). When these updates are introduce vertex, introduce edge, or forget vertex nodes, it takes only $O(1)$ work to calculate each new value. However, at a join node $t$, one must take a maximum across all $b_1\leq b$ for each state. Thus, for every labeling $f$, we solve an instance of MCKP to calculate values $V(t,b,f)$ for all budgets $b\leq r$, requiring time $O(r^2)$. Across all labelings, our join node updates will therefore take at most $O(2^k\cdot r^2)$ time, giving a total runtime for \ourAlgorithm  of $O(k2^k\cdot nr^2)$ in the worst case. When treewidth $k$ is bounded by a constant, the algorithm runs in $O(nr^2)$ time and is FPT in parameter $k$. Furthermore, \ourAlgorithm  is FPL in parameters $(k,r)$. This performance matches that of \myAlgorithm, as trees have treewidth $1$. 

As with solving REIC on trees, we can reconstruct our solution by maintaining a predecessor array $P(t,b,f)$ of pointers from every state of inner node $t$ to states for each child of $t$, as well as a similar indicator array $I(t,b,f)$. Solution reconstruction by a post-order traversal of tree decomposition $T$ following our predecessor array $P$ takes time $O(kn)$, which is dominated by the algorithm's runtime. Storing both the dynamic program table and solution reconstruction requires $O(k2^k\cdot nr)$ space.

\subsubsection{Comparison to Courcelle's Theorem}\label{sec:courcelle}
To complete the analysis of our fixed-parameter linear DP algorithms for REIC, we argue that our explicit constructions demonstrate significant advantages over Courcelle's theorem \cite{Courcelle1990}, a meta-algorithm for solving problems definable in monadic second-order logic (MSO) on graphs with bounded treewidth in linear time.
In \cref{sec:mso-alg}, we show that REIC can be formulated as a linear extended MSO extremum problem and hence solved in $O(n)$ time on rational-weighted graphs of bounded treewidth when budget $r$ is constant. 
Notably, both Courcelle's theorem and our DP algorithms are FPL in parameters~$(k,r)$, where $k=\tw(G)$.
However, we give two distinct advantages of our DP algorithms for REIC over the meta-algorithm.

On the theoretical side, Courcelle's theorem can have unbounded dependence on formula size, which includes the budget constraint~$r$.
Thus, Courcelle's theorem does not give an FPT algorithm for REIC for any $r=\omega(1)$ and so is constrained to instances where $r$ is constant;
in contrast, our DP algorithms achieve time complexity at worst $O(n^3)$ when treewidth $k$ is bounded, regardless of the value of $r$.

This dependence of Courcelle's theorem on formula size also leads to a convincing advantage in practical efficiency for our DP algorithms, as the size of the tree automaton construction suffers from extremely large hidden constants in practice.
As noted by Kneis and Langer, these constants are the barrier to there being a practical implementation of Courcelle's theorem
\cite{KneisLanger09}.  
For some intuition on this blow-up of automaton size, \cite{KneisLangerRossmanith11} observes that each quantifier alternation requires a power set construction for the automaton; in our case, the formula for REIC-sat involves two quantifier alternations (from $\forall$ to $\exists$ to $\forall$).

In contrast to Courcelle's theorem, our DP approach unconditionally offers a direct and computationally efficient algorithm for this problem. We demonstrate our approach's empirical effectiveness in our numerical study of \myAlgorithm in \cref{sec:experimental}, and while we did not implement \ourAlgorithm for bounded treewidth graphs, we expect it to be similarly performant due to the shared DP structure and FPL complexity with our algorithm for trees.

\section{The \texorpdfstring{$r$}{r}-Facility Interdiction Covering Problem}
\label{sec:facility}

A natural variant of REIC is to consider facility removal as opposed to edge removal, as in the original $r$-interdiction covering problem for facility location~\cite{ChurchSM04}. 
In the graph setting, we arrive at the $r$-facility interdiction covering problem~(RFIC), where an interdiction strategy for the same covering objective consists of a subset of facilities $R\subset S$ subject to budget constraint $|R|\leq r$.
Letting $C_G(v,R)=1$ indicate that customer $v \in V\setminus S$ can reach some facility $s\in S\setminus R$ in the graph $G\setminus R$, we define the problem.
\begin{definition}[RFIC]\label{def:rfic}
    In the $r$-facility interdiction covering problem, we are given a graph $G=(V,E)$ with facilities $S\subset V$, customer weights $w:V\setminus S\rightarrow \mathbb{R}_{\geq 0}$, and an integer budget $r<|S|$, with the goal of selecting a subset $R\subset S$ with $|R|\leq r$ maximizing $\sum_{v\in V\setminus S}w(v)(1-C_G(v,R))$.
\end{definition}
Modifying our DP approach for REIC from \myAlgorithm, we give an $O(nr^2)$-time algorithm for RFIC on trees.
To complement this FPL result for trees, we show that RFIC on general graphs is W[1]-hard by a parameterized reduction from \textsc{Clique}.

We also prove hardness of approximation results for RFIC, which is complicated by the fact that the problem has two natural approximation objectives, depending on whether interpreted as (i) a minimization problem for the weight of covered customers (RFIC-min), or (ii) a maximization problem for the weight of disconnected customers (RFIC-max).
We give an approximation-preserving reduction from small set bipartite vertex expansion (SSBVE) to RFIC-min, observing that RFIC generalizes SSBVE to a weighted setting.\footnote{Either of these reductions can be taken to show that RFIC is NP-hard.} 
This gives an $O(n(n-k)^2)=O(n^3)$-time algorithm for SSBVE on trees, while also proving RFIC-min is as hard to approximate as small set bipartite vertex expansion (and hence smallest $p$-edge subgraph).
We also note that our parameterized reduction from \textsc{Clique} shows that RFIC-max is as hard to approximate as densest $k$-subgraph.

\subsection{DP Algorithm for RFIC on Trees}\label{sec:facility-DP}
Following \myAlgorithm, we describe a DP algorithm for solving RFIC on trees which traverses the tree upwards from leaves and maximizes the total weight of customers disconnected within the current subtree. Again, we build upon conditions $XY$; fixing node $v$ with parent $w$, assumption $X=0$ promises that no facility is reachable from $v$ through $w$, and guarantee $Y=0$ requires that no facility is reachable from $v$ in the subtree of $v$. \(V_{XY}(v, b)\) denotes the maximum objective achievable in the subtree rooted at \(v\) under conditions \(XY\) and with node removal budget \(b\). We begin with our base case by defining values $V_{XY}(v,b)$ for leaf nodes.

\subparagraph*{Leaves}

Let \(v\) be leaf of the tree.

\begin{itemize}
\item For a facility leaf node, we must spend one unit of budget and remove $v$  to enforce $Y=0$, or otherwise set $Y=1$; $X$ does not affect this. Thus, for all $b\leq r$,
  \[
    V_{XY}(v,b) = \begin{cases}
      0 & \text{if } Y=1 \text{ or } b\ge 1,\\
      -\infty & \text{otherwise}.
    \end{cases}
  \]
\item For a customer leaf node, we can gain value $w(v)$ by guaranteeing $X=0$, and we cannot enforce $Y=1$. Therefore, for all $b\leq r$,
  \[
    V_{XY}(v,b) = \begin{cases}
      -\infty & \text{if } Y=1,\\
      w(v) & \text{if } Y=0 \text{ and } X=0,\\
      0 & \text{otherwise}.
    \end{cases}
  \]
  \end{itemize}
With DP values calculated for each leaf, we now consider an inner node \(v\) and its children \(u_1, \ldots, u_k\), calculating \(V_{XY}(v,b)\) for all budgets \(b \leq r\) and cases \(XY \in \{00,01,10,11\}\) by solving a multiple choice knapsack problem.

\subparagraph*{Inner Nodes: Case $V_{X0}$}

If $v$ is a facility, we must remove $v$ to enforce $Y=0$, regardless of condition $X$. Thus, for all $b\leq r$,
\begin{align}
V_{X0}(v,b) = \max_{b_i:\sum_{i=1}^k b_i = b-1} \left\{ \;\sum_{i=1}^k\;\max\left\{
V_{00}(u_i,b_i),
V_{01}(u_i,b_i)
\right\} \right\} \,.
\end{align}
If $v$ is a customer and $X=0$, then in order to maintain $Y=0$, we must require that every child $u_i$ has $Y=0$. Thus, for all $b\leq r$,
\begin{align}
V_{00}(v,b) = w(v) + \max_{b_i:\sum_{i=1}^k b_i = b}
    \left\{ \sum_{i=1}^k 
        V_{00}(u_i, b_i)
     \right\}\,.
\end{align}
If $v$ is a customer and $X=1$, then in order to maintain $Y=0$, we must require that every child $u_i$ has $Y=0$. Thus, for all $b\leq r$,
\begin{align}
V_{10}(v,b) = \max_{b_i:\sum_{i=1}^k b_i = b}
    \left\{ \sum_{i=1}^k 
        V_{10}(u_i, b_i)
     \right\}\,.
\end{align}

\subparagraph*{Inner Nodes: Case $V_{X1}$}

\begin{align}\label{eq:rfic_V11_facility}
V_{X1}(v,b) = \max_{b_i:\sum_{i=1}^k b_i = b} \left\{ \;\sum_{i=1}^k\;
\max\left\{
V_{10}(u_i,b_i),
V_{11}(u_i,b_i)
\right\} \right\}\,.
\end{align}
When $v$ is a customer, we must again enforce that $v$ is connected to a facility in its subtree with additional variables $\{z_i\}_{i=1}^{k}\subseteq \{0,1\}^k$ such that $\sum_{i=1}^{k} z_i\geq 1$. Now, when $v$ is a customer, we have that for all $b\leq r$,
\begin{align}\label{eq:rfic_V11_customer}
V_{X1}(v,b) = \max_{b_i,z_i:\sum_{i=1}^k b_i = b, \sum_{i=1}^k z_i \geq 1}
\Bigg\{ \;\sum_{i=1}^k\;  (1-z_i) V_{10}(u_i,b_i)
+ z_i V_{X1}(u_i,b_i)\Bigg\}\,.
\end{align}

\subparagraph*{Time Complexity and Solution Reconstruction}
As in \myAlgorithm, this algorithm for node removal incurs cost $O(kr^2)$ at every inner node $v$ by solving a multiple-choice knapsack instance, where $k$~is the number of children of~$v$. By the same charging argument, the time spent across all inner nodes is $O(nr^2)$, dominating the overhead for instantiation of leaves and explicit reconstruction of the optimal interdiction strategy by post-order traversal of predecessor and indicator arrays. Thus, this algorithm runs in time $O(nr^2)$.

\subsection{Reduction from SSBVE to RFIC}
\label{sec:SSBVE}
We begin by defining SSBVE, where $N_G(S)$ denotes the neighborhood of $S$ in $G$. We drop the subscript when $G$ is clear from context.
\begin{definition} [SSBVE]
    In the small set bipartite vertex expansion problem, we are given a bipartite graph $G=(U\cup V, E)$ and an integer $k<|U|$, with the goal of selecting a subset $S\subset U$ with $|S|=k$ minimizing $|N_G(S)|$.
 \end{definition}
 An algorithm is an $\alpha$-approximation for SSBVE if it returns a set $S$ with $|N(S)|\leq \alpha|N(S^*)|$, where $S^*$ is the optimal set. 
 Now, we give an approximation-preserving reduction from SSBVE to the \emph{minimization} interpretation of $r$-facility interdiction covering, where the objective is to minimize (subject to the budget constraint) the total value of customers covered by remaining facilities $S\setminus R$, denoted $\gamma(R):=\sum_{v\in V\setminus S}w(v)C_G(v,R)$.
\begin{theorem}
   If there exists a polynomial-time $\alpha(|S|)$-approximation algorithm for RFIC-min, then there exists a polynomial-time $\alpha(|U|)$-approximation for SSBVE.
\end{theorem}
\begin{proof}
    Let $G=(U\cup V, E)$ be a bipartite graph, and let $I=(G,k)$ be an instance of SSBVE. We define $G'$ to be the graph $G$ with facility set $S=U$ and unit customer weights $w(v)=1$ for all $v\in V$. Thus, we construct an RFIC instance $I'=(G',|U|-k)$, where the budget $r=|U|-k$. Now, let $\mathcal{A}$ be an $\alpha$-approximation for RFIC, so that $\mathcal{A}(I')$ returns an interdiction strategy $R\subseteq U, |R|=|U|-k$ with objective value $\gamma(R)\leq \alpha \gamma(R^*)$, where $R^*$ is the optimal interdiction strategy. We observe that $v\in V$ is covered by remaining facilities $U\setminus R$ if and only if $v\in N(U\setminus R)$. Thus, $\gamma(R)=\sum_{v\in N(U\setminus R)}w(v)=|N(U\setminus R)|$. Now, we claim that outputting the set $U\setminus R$ gives an $\alpha$-approximation for $I$. First, it is easy to see that $|U\setminus R|=k$. Next, we use the fact that $|N(U\setminus R)|=\gamma(R)\leq \alpha \gamma(R^*)\leq \alpha |N(U\setminus R^*)|$. Since $R^*\subseteq U$ is a set of $|U|-k$ facilities minimizing the weight of covered customers by $U\setminus R^*$, then $U\setminus R^*$ is a set of $k$ vertices $Z\subseteq U$ minimizing $|N(Z)|$. Thus, $U\setminus R$ is an $\alpha$-approximation instance $I$ of SSBVE. This reduction is linear-time.
\end{proof}
We note that this does not alter the structure of $G$ and hence gives an exact algorithm for SSBVE on trees by employing our exact algorithm (i.e., our 1-approximation) for RFIC on trees.\footnote{By this observation, a similar extension of our DP algorithm for REIC on bounded treewidth graphs to RFIC would give an algorithm for SSBVE on bounded treewidth graphs.} 
The key step in this reduction is observing that $v$ is covered by facilities $U\setminus R$ when $v\in N(U\setminus R)$, and hence $\gamma(R)=\sum_{v\in N(U\setminus R)} w(v)$. 
This expression captures the fact that RFIC is a ``weighted version'' of SSBVE, but does not extend beyond this setting where a facility can only cover customers in its neighborhood. 
However, we make this intuition explicit by proving that any instance of RFIC is equivalent to a bipartite instance of RFIC with partition $U=S$, which we call bip-RFIC.

\begin{lemma}\label{lem:bipartite}
For any instance of RFIC, there is an equivalent instance of bip-RFIC.
\end{lemma}

\begin{proof}
    Let $G=(V,E)$ be a graph with facilities $S$ and customer weights $w$. Now, consider the components of $G\setminus S$, which we label $\mathcal{C}_1,\ldots,\mathcal{C}_d$. Define a set of customers $C=\{v_i:i\in [d]\}$ with weight function $w':C\rightarrow \mathbb{R}_{\geq 0}$ given by $w'(v_i)=\sum_{v\in \mathcal{C}_i}w(v)$. Lastly, define a new set of edges $E'=\{(N_G(\mathcal{C}_i)\cap S)\times\{v_i\}:i\in [d]\}$, which connects every vertex $v_i$ to each facility neighboring a customer in $\mathcal{C}_i$. Then we claim solving RFIC on the bipartite graph $G'=(S\cup C, E')$ with facilities $S$ and customer weights $w'$ is equivalent to solving RFIC on the original graph $G$. To see this, we will show that every interdiction strategy $R\subset S$ achieves the same objective value on $G$ and $G'$. 

    Here, we use the minimization interpretation of RFIC, letting $\gamma_G(R)$ denote the objective value of strategy $R$ in the original graph and $\gamma_{G'}(R)$ denote the objective value of $R$ in the bipartite graph. First, we observe that $v_i$ is only disconnected if each facility in $S_i:=N_G(\mathcal{C}_i)\cap S$ is removed, i.e., if $S_i\subseteq R$. Thus, $\gamma_{G'}(R)=\sum_{i:S_i\nsubseteq R} w'(v_i)$. Now, consider some customer in the original graph $v\in \mathcal{C}_i$. If $S_i\subseteq R$, then every facility adjacent to any customer reachable from $v$ by a path of customers is interdicted, and so $v$ cannot reach a facility. On the other hand, assume $S_i\nsubseteq R$, i.e., there is a facility $s\in S_i, s\notin R$. Since $s\in S_i$, there must be some customer $u\in \mathcal{C}_i$ s.t. $s$ is adjacent to $u$. As $u$ and $v$ are both in $\mathcal{C}_i$, they are connected in the graph $G\setminus S$ and hence in the graph $G\setminus R$. Thus, by concatenation we obtain the path $v,\dots,u,s$ connecting $v$ to a facility, and hence $v$ is still covered under interdiction strategy $R$. Therefore, the objective value of $R$ in the original graph is $\gamma_{G}(R)=\sum_{i:S_i\nsubseteq R}\sum_{v\in \mathcal{C}_i}w(v)=\gamma_{G'}(R)$, proving the claim.
\end{proof}

\subsection{W[1]-Hardness of RFIC}\label{sec:hardness}
We show a reduction to $r$-facility interdiction covering from $\textsc{Clique}$, a canonical W[1]-complete problem. For the decision version of RFIC, an instance $(G,r,k)$ is a yes-instance if there is an interdiction strategy $R$ subject to budget constraint $|R|\leq r$ with interdiction value $k$.

\begin{theorem}
    The decision version of RFIC is $\textup{W}[1]$-hard.
\end{theorem}
\begin{proof}
   To show that the problem is $\textup{W}[1]$-hard, we give a parameterized reduction from \textsc{Clique}. An instance $(G, k)$ of \textsc{Clique} is a yes-instance if $G=(V,E)$ contains the clique on $k$ vertices $K_k$ as an induced subgraph, and a no-instance otherwise. 
   
   We now define graph $G'$, such that $G$ has a $k$-clique if and only if $G'$ has an interdiction strategy of budget $k$ with value $\binom{k}{2}$. Let $V_E=\{v_e:e\in E\}$ be a set of vertices `bisecting' every edge of $G$, such that edges $E'=\{(v_e,a):e=(a,b)\in E\}$ connect new vertices $v_{(a,b)}\in V_E$ to the original vertices $a,b\in V$. Let $G'=(V\cup V_E, E')$, with facilities $S=V$ and unit customer weights $w(v_e)=1$. Observe that for edge $e=(a,b)$, every path from customer $v_e$ to a facility goes through $a$ or $b$. Thus, $v_e$ is disconnected if and only if $\{a,b\}\subseteq R$. Equivalently, a customer $v_e$ is disconnected if and only if edge $e$ is in the induced subgraph $G(R)$. By our budget constraint, this subgraph has at most $k$ vertices, so the instance of RFIC has value $\binom{k}{2}$ if and only if there is a set $R$ such that  $G(R)$ has $\binom{k}{2}$ edges, i.e., $G(R)=K_k$. Therefore, $(G,k)$ is a yes-instance if and only if $(G,k,\binom{k}{2})$ is a yes-instance, completing the reduction. 
\end{proof}

We can also consider the \emph{maximization} version of $\textsc{Clique}$, the densest $k$-subgraph problem (DkS), to show hardness of approximation for the maximization interpretation of RFIC, as presented in \cref{def:rfic} but with facility removal. The following result follows from the same reduction as before.
\begin{corollary}
   If there exists a polynomial-time $\alpha(|S|)$-approximation algorithm for RFIC-max, then there exists a polynomial-time $\alpha(n)$-approximation for DkS, where $n$ is the number of nodes in the graph for the DkS instance.
\end{corollary}

\section{Conclusions and Future Directions}\label{sec:conclusions}
We provided an efficient solution for determining edge removals that prevent the most value of customers from reaching facilities in trees and bounded treewidth graphs. Our numerical results on trees showed that our algorithm is significantly faster than state-of-the-art baselines. Further study by evaluation on real-world datasets could be informative, such as in the motivating satellite networks application. 
We also considered a variant where facilities are removed, proving it to be a generalization of small set bipartite vertex expansion and densest $k$-subgraph. Future work could study whether known algorithms for these problems generalize to RFIC, or if this weighted generalization is fundamentally more difficult. 
Lastly, our study invites the question of whether these interdiction problems are tractable on other sparse families of graphs; interdicting the covering objective relies on disconnection in edge and facility removal variants, and so separator techniques for planar graphs might be applicable.


\newpage 
\bibliographystyle{plainurl}
\bibliography{main}

\begin{appendices}
\crefalias{subsection}{appendix}
\section{Constrained Multiple-Choice Knapsack Problem} 
\label{sec:MCKP}
In this section, we describe the constrained multiple-choice knapsack problem~(CMCKP) and give an $O(NC^2)$-time dynamic programming algorithm based on the algorithm for MCKP \cite[Section 11.5]{Kellerer04}. In the multiple choice variant of the classic knapsack problem, we are packing items from $k$ mutually disjoint classes, or ``buckets,'' into a knapsack of capacity $C$, with the additional constraint that we must select exactly one item from each bucket. For the ``constrained'' problem, a subset of items satisfy a given property of interest, and we must select at least one item satisfying the property.

For notation, we define buckets $i= 1, \dots, k$ where bucket $i$ contains $N_i > 0$ items. We use the following notation for $i=1, \dots, k$ and $j=1,\dots,N_i$:
\begin{itemize}
    \item $x_{ij}\in \{0,1\}$ represents whether item $j$ is included from bucket $i$;
    \item $c_{ij}\geq 0$ and $v_{ij}\geq 0$ represent the cost and value accrued by including such item;
    \item $z_{ij}\in \{0,1\}$ represents whether item $j$ in bucket $i$ satisfies a property of interest.
\end{itemize}

For this knapsack problem, we want to maximize the total profit given a maximum capacity $C$, subject to requiring in our knapsack exactly one item from each bucket and at least one item satisfying the property of interest.
In our case, we actually want the optimal profit for all capacities $0\leq c\leq C$, which is naturally given as the output of the following DP formulation.
For some budget $0 \leq c \leq C$, consider the subproblem $V_Z(t,c)$ restricted to the first $1 \leq t \leq k$ buckets, where the ``property of interest'' constraint is satisfied only when $Z=1$,
\begin{align*}
V_0(t,c) := &\max\left\{\sum_{i=1}^t\sum_{j=1}^{N_i} v_{ij}x_{ij}: \sum_{i=1}^t\sum_{j=1}^{N_i} c_{ij}x_{ij} \leq c, \sum_{j=1}^{N_i} x_{ij} = 1 \;\forall i, \sum_{i=1}^t\sum_{j=1}^{N_i} z_{ij}x_{ij} = 0\right\}\,,\\
V_1(t,c) := &\max\left\{\sum_{i=1}^t\sum_{j=1}^{N_i} v_{ij}x_{ij}: \sum_{i=1}^t\sum_{j=1}^{N_i} c_{ij}x_{ij} \leq c, \sum_{j=1}^{N_i} x_{ij} = 1  \;\forall i, \sum_{i=1}^t\sum_{j=1}^{N_i} z_{ij}x_{ij} \geq 1\right\}\,.
\end{align*}

We can iteratively calculate the DP values $V_Z(t,\cdot)$ for $Z\in\{0,1\},t=1,\ldots, k$, ultimately returning the set of optimal knapsacks $V_1(t,\cdot )$, which must satisfy constraint $Z=1$. This is explicitly described in \cref{alg:MCKP}, and we can reconstruct our solution by maintaining a simple predecessor array as described in \cref{sec:reconstruction}. Both solving and reconstructing our DP solutions can be done in time $O(NC^2)$, where $N=\sum_{i=1}^{k}N_i$ is the total number of items.

\SetKwProg{Fn}{Function}{}{end}%
\SetKwFunction{CMCKP}{CMCKP}%

\begin{algorithm}[t]
\caption{Constrained Multiple-Choice Knapsack Problem}\label{alg:MCKP}
\DontPrintSemicolon
\Fn{\CMCKP{\texttt{costs}, \texttt{values}, \texttt{properties}, $C$}}{
\SetAlgoLined
\SetKwInOut{KwIn}{Input}
\SetKwInOut{KwOut}{Output}

\KwIn{arrays \texttt{costs}, \texttt{values}, and \texttt{properties} indexed by bucket $1,\ldots,k$ and item $1,\ldots,N_i$ in bucket $i$, and an integer knapsack capacity $C$}
\KwOut{optimal solution values of CMCKP for capacities $0\leq c\leq C$}

$V_0,V_1\gets -\infty\in\mathbb{R}^{(k+1)\times(C+1)}$\;
$V_0[0,\cdot\ ]\gets 0$\;

\For{$i\gets1$ \KwTo $k$}{
  \For{$j\gets1$ \KwTo $N_i$}{
    $c_{ij}\gets\texttt{costs}[i][j]$, $v_{ij}\gets\texttt{values}[i][j]$, $p_{ij}\gets\texttt{properties}[i][j]$\;
    \For{$c\gets 0$ \KwTo $C-c_{ij}$}{
        $x_Z\gets V_Z[i-1,c+c_{ij}]+v_{ij}$\;
      \eIf{$p_{ij}=0$}{
        $V_Z[i,\,c]\gets \max\bigl(V_Z[i,\,c],\,x_Z\bigr)$
      }{
        $V_1[i,\,c]\gets \max\bigl(V_1[i,\,c],\,x_0,\,x_1\bigr)$
      }
    }
  }
}

\Return{$V_1[k,\cdot \ ]$}  
}
\end{algorithm}

\section{Courcelle's Theorem and MSO Formulation of REIC}\label{sec:mso}

In this section, we give a succinct treatment of MSO logic and linear extended MSO extremum problems in order to present \cref{thm:linemsol}, an extension of Courcelle's theorem for linear optimization problems introduced by Arnborg, Lagergren, and Seese \cite{ArnborgLagergrenSeese91}. We then give a satisfaction formulation of REIC in MSO, allowing us to define REIC as a linear extended MSO extremum problem; by \cref{thm:linemsol}, we obtain a linear time algorithm for REIC on rational-weighted graphs of bounded treewidth when the budget $r$ is fixed. We loosely follow the notation of \cite{BurtonDowney17Triangulations}, while also incorporating evaluation functions as described in  \cite{ArnborgLagergrenSeese91,CourcelleMakowskyRotics98}.

\subsection{MSO Logic and Courcelle's Theorem for Optimization}\label{sec:mso-background}
Monadic second-order logic (MSO) is the fragment of second-order logic where second-order quantification is limited to quantification over sets. For MSO on graphs, we use what is formally known as \emph{extended} MSO logic, which allows direct access to variables and sets of edges. MSO logic supports standard boolean operations, such as $\wedge,\vee,\neg$, and $\rightarrow$, as well as variables for nodes, edges, sets of nodes, and sets of edges. Upon these variables, MSO supports $\in,=,$ incidence, and adjacency binary relations, as well as $\forall$ and $\exists$ quantification over all variables.

We use the notation $\varphi(x_1,\ldots,x_t)$ to denote an MSO formula with $t$ free variables, which are placeholder variables not bound by $\forall$ or $\exists$ quantifiers. An MSO sentence is an MSO formula with no free variables. For simple graph $G$ and MSO sentence $\varphi$, the notation $G\vDash \varphi$ signifies that the interpretation of $\varphi$ in graph $G$ is a true statement.

We now define a linear extended MSO extremum problem, restricting our class of structures to be the class of simple graphs. For a set variable $X$, we use $|X|_j$ to denote the evaluation $\sum_{a\in X}f_j(a)$ for evaluation function $f_j$.
\begin{definition}
    A \emph{linear extended MSO extremum problem} $P$ consists of an MSO formula $\varphi(A_1,\ldots,A_t)$ with free set variables $A_1,\ldots,A_t$, evaluation functions $f_1$,..., $f_m$ associating rational values to the elements of $G$, and a rational linear function $g(x_1,\ldots,x_{tm})$, for constants $t$ and $m$. Its interpretation is as follows: given a simple graph $G$ as input, we are asked to maximize $g(|A_1|_1, \ldots,|A_t|_m)$ over all sets $A_1,\ldots,A_t$ for which $G\vDash \varphi(A_1,\ldots,A_t)$.
\end{definition}
We now restate a powerful extension of Courcelle's theorem for such optimization problems.
\begin{theorem}[Arnborg, Lagergren and Seese \cite{ArnborgLagergrenSeese91}]\label{thm:linemsol}
For any linear extended MSO extremum problem $P$ and any class $K$ of simple graphs with universally bounded treewidth, it is possible to solve $P$ for graphs $G\in K$ in time $O(|G|)$ under the uniform cost measure.\footnote{The uniform cost measure assumes that all elementary arithmetic operations run in constant time.}
\end{theorem}
In particular, when $\tw(G)$ is the treewidth of graph $G$, then \cref{thm:linemsol} gives a time complexity of $O(f(k,|\varphi|)\cdot |G|)$, where $f$ is some computable function.

\subsection{MSO Formulation of REIC}
\label{sec:mso-alg}

We begin by giving a simple MSO formula $\text{Card}_{=p}(X)$ expressing that a set $X$ has cardinality $p$ for any nonnegative integer $p$. First, consider the following formula for $|X|\geq p$,
\begin{align*}
\text{Card}_{\geq p}(X)\equiv \exists& a_1\in X\ \exists a_2\in X\ldots\  \exists a_p\in X\\
&\Bigl[(a_1\neq a_2)\wedge (a_1\neq a_3)\wedge\cdots \wedge (a_{p-1}\neq a_p)\Bigr] \,.
\end{align*}
We can therefore express $|X|=p$ by,
\[
\text{Card}_{=p}(X)\equiv \text{Card}_{\geq p}(X)\wedge \neg \text{Card}_{\geq p+1}(X)\,.
\]
We now define the satisfaction formulation for $r$-edge interdiction covering. Let $G=(V,E)$ be a simple graph with facilities $S$ and customers $V\setminus S$ carrying weight $w:V\setminus S\rightarrow \mathbb{Q}_{\geq 0}$. For budget $r$, interdiction strategy $R$, and disconnected customers $Y$, we define REIC-sat as follows:
\begin{boxtext}
{\emph{Is $R$ a set of at most $r$ edges such that, after deleting $R$ from $G$, each customer $v\in Y$ is disconnected from every facility?}}
\end{boxtext}
Recall that $C_G(v,R)=1$ asserts that $v$ is connected to some facility after $R$ is removed from $G$, and $C_G(v,R)=0$ otherwise. 
To express REIC-sat in MSO, we first define a formula for $C_G(v,R)$, building upon a simple MSO formula for connected component,
\begin{align*}
  C_G(v,R)\equiv \exists &s\in S\,\exists X\subseteq V\Bigl( v\in X\wedge s\in X \wedge \\&\Bigl[\forall u\in X\,\forall w\in V\bigl((u,w)\in E\setminus R\to w\in X\bigr) \Bigr]\Bigr)\,.
\end{align*}
Specifically, we check whether there exists a set of vertices \(X \subseteq V\) such that: (i) \(v \in X\) (ii) \(s \in X\) for some facility $s\in S$; (iii) \(X\) is a maximal connected component. If such a set \(X\) exists, then $v$ is still connected to a facility in $G\setminus R$, so \(C_G(v,R) = 1\); otherwise, there is no facility in the connected component of $v$ in $G\setminus R$, so $C_G(v,R)=0$. 

We now give our MSO formulation for REIC-sat by the formula $\varphi_r(R,Y)$ with free variables $R$ and $Y$, for fixed $r\in\mathbb{Z}^+$,
\[ 
  \varphi_{r}(R,Y)\equiv\Bigl(\text{Card}_{=r}(R)\Bigr)\;\wedge\;\Bigl(\forall v\in Y \left(\neg C_G(v,R)\right)\Bigr)\,.
\]
To complete our definition of REIC as a linear extended MSO extremum problem, we define evaluation function $f_1$ as well as linear function $g$,
\[
f_1=\begin{cases}
    w(v),&\text{if }a=v\in V\setminus S\\
    0, &\text{o.w.}
\end{cases}\quad;\quad g(|R|_1,|Y|_1)=|R|_1+|Y|_1\,.
\]
By our definition of $f_1$, we have that $g(|R|_1,|Y|_1)=\sum_{v\in Y}w(v)$; as $Y$ is a set of disconnected customers under interdiction strategy $R$, we have our covering objective function as desired. Applying \cref{thm:linemsol}, for any graph $G$ of bounded treewidth and fixed budget $r$, we can maximize the covering objective function of $Y$ subject to $R$ disconnecting customers $Y$ and $|R|=r$ in linear time. Therefore, when $r$ is fixed, we have a linear-time algorithm for REIC on rational-weighted graphs of bounded treewidth.


\subparagraph*{Dependence on $r$:} The MSO formula $\varphi_r$ for REIC-sat relies on the formula $\text{Card}_{=r}$ for verifying the interdiction strategy satisfies the budget constraint. 
While $|\text{Card}_{=r}|$ has quadratic dependence on $r$ in our naïve construction, it is clear that any construction for $\text{Card}_{=r}$, and hence for $\varphi_r$, has formula size dependent on~$r$.
\end{appendices}

\end{document}